\documentclass[11pt,letterpaper]{article}

\usepackage{etoolbox}

\usepackage[margin=1in]{geometry}

\usepackage[utf8]{inputenc}

\usepackage{amsmath}
\usepackage{amssymb}
\usepackage{amsthm}
\usepackage{bbm}
\usepackage{thm-restate}
\newtheorem{theorem}{Theorem}
\newtheorem{lemma}[theorem]{Lemma}

\newtheorem{definition}[theorem]{Definition}

\newtheorem{observation}[theorem]{Observation}

\usepackage{url}
\usepackage[table]{xcolor}
\usepackage{array}

\usepackage{setspace}

\usepackage{wrapfig}

\makeatletter
\newcommand{\labeltarget}[1]{\Hy@raisedlink{\hypertarget{#1}{}}}
\makeatother

\usepackage{lmodern}
\usepackage{times}
\usepackage{upgreek}

\usepackage{complexity}

\usepackage[inline]{enumitem}
\usepackage{moreenum}

\usepackage{mathtools}
\usepackage[super]{nth}

\usepackage{bm}

\usepackage{xspace}

\usepackage[immediate]{silence}
\usepackage{sectsty}
\DeactivateWarningFilters[tmp]
\allsectionsfont{\boldmath}

\setlist[enumerate]{nosep,topsep=0em}
\setlist[enumerate,1]{label=(\roman*), leftmargin=2.2em}
\setlist[itemize]{nosep,topsep=0.1em}

\usepackage[textsize=footnotesize, color=blue!30!white]{todonotes}

\usepackage{xfrac}

\usepackage{tikz}
\usetikzlibrary{calc}
\usetikzlibrary{shapes.geometric}
\usetikzlibrary{arrows}
\usetikzlibrary{decorations.pathreplacing, decorations.pathmorphing}

\usepackage{tcolorbox}
\tcbuselibrary{skins}

\usepackage{float}
\usepackage{graphicx}
\graphicspath{{graphics/}}
\makeatletter
\newcommand\appendtographicspath[1]{%
  \g@addto@macro\Ginput@path{#1}%
}
\makeatother

\usepackage[margin=10pt,font=small,labelfont=bf,skip=-5pt]{caption}
\usepackage{subcaption}

\usepackage[
backend=biber,
style=alphabetic,
citestyle=alphabetic,
maxcitenames=99,
mincitenames=98,
maxbibnames=99,
]{biblatex}

\usepackage{xpatch}

\makeatletter
\patchcmd\blx@bblinput{\blx@blxinit}
                      {\blx@blxinit
                      }{}{\fail}
\makeatother

\usepackage[vlined,ruled,algo2e]{algorithm2e}
\SetAlgoSkip{bigskip}
\SetAlgoInsideSkip{smallskip}

\usepackage{mdframed}

\usepackage{bigfoot}
\usepackage[pdfpagelabels,bookmarks=false]{hyperref}
\hypersetup{colorlinks, linkcolor=darkblue, citecolor=darkgreen, urlcolor=darkblue}

\definecolor{darkblue}{rgb}{0,0,0.38}
\definecolor{darkred}{rgb}{0.8,0,0}
\definecolor{darkgreen}{rgb}{0.1,0.35,0}

\DeclareMathOperator{\apex}{apex}
\DeclareMathOperator{\argmin}{argmin}
\DeclareMathOperator{\slack}{slack}

\newcommand\OPT{\ensuremath{\mathrm{OPT}}}
\newcommand\Drop{\ensuremath{\mathrm{Drop}}}
\renewcommand{\epsilon}{\varepsilon}

\makeatletter
\let\@@pmod\pmod
\DeclareRobustCommand{\pmod}{\@ifstar\@pmods\@@pmod}
\def\@pmods#1{\mkern8mu({\operator@font mod}\mkern 6mu#1)}
\makeatother

\makeatletter
\let\@@mod\mod
\DeclareRobustCommand{\mod}{\@ifstar\@mods\@@mod}
\def\@mods#1{\mkern8mu{\operator@font mod}\mkern 6mu#1}
\makeatother

\def\Cscr{\mathcal{C}}

\makeatletter
\def\@fnsymbol#1{\ensuremath{\ifcase#1\or *\or %
\ddagger\or
    \mathsection\or \mathparagraph\or \|\or **\or \dagger\dagger
    \or \ddagger\ddagger \else\@ctrerr\fi}}
\makeatother

\title{A Better-Than-2 Approximation for Weighted Tree Augmentation\thanks{
This project received funding from Swiss National Science Foundation grant 200021\_184622 and the European Research Council (ERC) under the European Union's Horizon 2020 research and innovation programme (grant agreement No 817750).
}} 

\author{
Vera Traub\thanks{
Department of Mathematics, ETH Zurich, Zurich, Switzerland.
Email: \href{mailto:vera.traub@ifor.math.ethz.ch}%
{vera.traub@ifor.math.ethz.ch}.
}
\and
Rico Zenklusen\thanks{
Department of Mathematics, ETH Zurich, Zurich, Switzerland.
Email: \href{mailto:ricoz@ethz.ch}%
{ricoz@ethz.ch}.}
}

\date{}

\begin{document}

\maketitle
\thispagestyle{empty}
\addtocounter{page}{-1}
\enlargethispage{-1cm}

\begin{abstract}
We present an approximation algorithm for Weighted Tree Augmentation with approximation factor $1+\ln 2 + \epsilon < 1.7$. This is the first algorithm beating the longstanding factor of $2$, which can be achieved through many standard techniques.
\end{abstract}
 
\begin{tikzpicture}[overlay, remember picture, shift = {(current page.south east)}]
\begin{scope}[shift={(-1.1,2.5)}]
\def\hd{2.5}
\node at (-2*\hd,0) {\includegraphics[height=0.5cm]{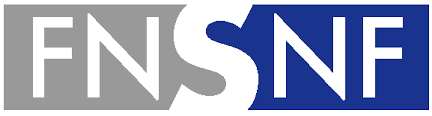}};
\node at (-\hd,0) {\includegraphics[height=1.0cm]{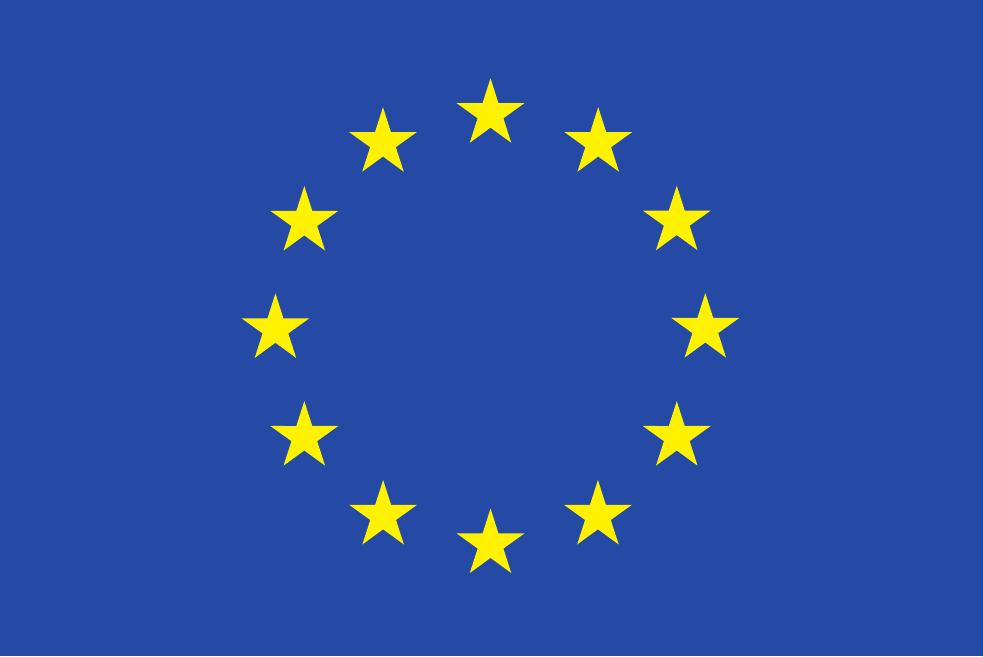}};
\node at (-0.2*\hd,0) {\includegraphics[height=1.2cm]{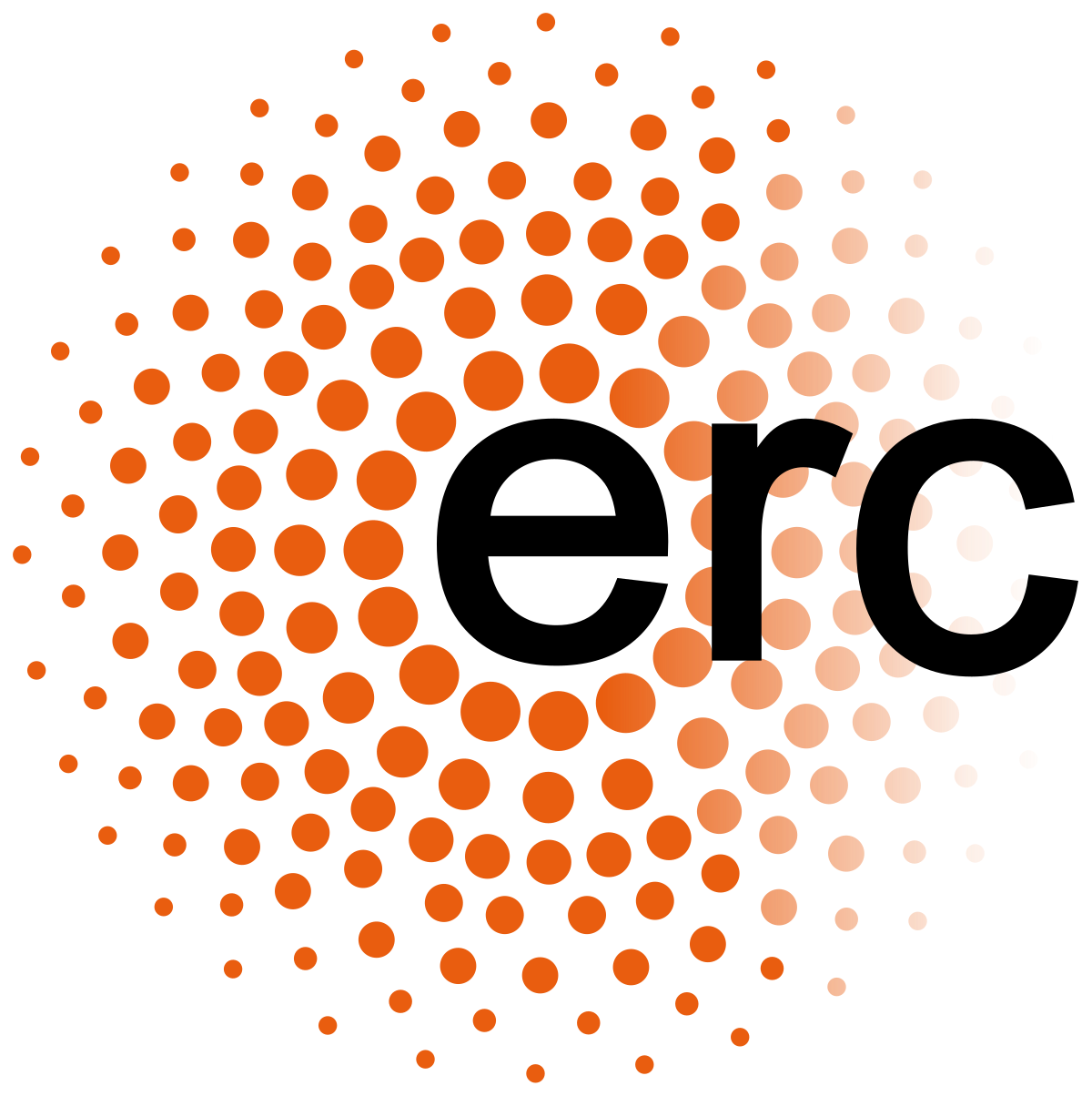}};
\end{scope}
\end{tikzpicture}

\clearpage

\section{Introduction}\label{sec:intro}

The Weighted Tree Augmentation Problem (WTAP) is among the most elementary and intensively studied connectivity augmentation problems. It asks how to increase the edge-connectivity of a graph from $1$ to $2$ in the cheapest possible way, and is formally described as follows. An instance consists of a spanning tree $G=(V,E)$ together with a set $L\subseteq \left(\!\begin{smallmatrix} V\\ 2 \end{smallmatrix}\!\right)$ of candidate edges to be added to $G$, which are also called \emph{links}, and positive link weights $w \colon L \to \mathbb{R}_{>0}$. The task is to find a minimum weight subset of links $F\subseteq L$ such that $(V,E\cup F)$ is $2$-edge-connected.%
\footnote{Depending on the literature, $0$-weight links may be allowed. This easily reduces to the case of strictly positive weights after including all $0$-weight links in the solution in a preprocessing step and continuing on the resulting reduced WTAP instance.}
It is easy to see that even though WTAP asks to augment the edge-connectivity of a spanning tree, it does capture the problem of increasing the edge-connectivity of an arbitrary $1$-edge-connected graph $G$ to $2$, because contracting all $2$-edge-connected components of $G$ leads to an equivalent WTAP instance. More generally, it is well-known that the problem of increasing the edge-connectivity of a graph from $k$ to $k+1$, for any odd $k$, reduces to WTAP (see, e.g.,~\textcite{cheriyan_1992_2-coverings}).

Already the unweighted version of WTAP, where all links have unit weight, is \NP-hard, even on trees of diameter~$4$, as shown by~\textcite{frederickson_1981_approximation}. The reduction they used was extended by~\textcite{kortsarz_2004_hardness} to show \APX-hardness of the unweighted version of WTAP.
Therefore, WTAP and special cases thereof have been heavily studied under the aspect of approximation algorithms. Prior to our work, the best approximation factor for WTAP was $2$, which was first shown by~\textcite{frederickson_1981_approximation} in the early '80s. They reduced the problem to finding a shortest arborescence, while losing a factor of $2$.
\citeauthor{frederickson_1981_approximation}'s procedure was subsequently simplified and significantly sped up by~\textcite{khuller_1993_approximation}. Moreover, many classical and very versatile techniques for network design problems developed later also lead to a $2$-approximation for WTAP.
This includes primal-dual approaches (see~\textcite{goemans_1994_improved}), the iterative rounding technique by~\textcite{jain_2001_factor}, and various further methods that are readily adaptable to WTAP (for example, a flow-based method by~\textcite{frank_1989_application} for certain directed connectivity problems; see also discussion in~\cite{khuller_1993_approximation} on how this relates to WTAP). 
However, despite extensive work on the problem and variations thereof, and some progress on special cases (see Section~\ref{sec:furtherRelatedWork}), the four decades old approximation factor of~$2$ by~\textcite{frederickson_1981_approximation} remained the state of the art.

\subsection{Our results}

In this paper, we present the first better-than-$2$ approximation for WTAP.
\begin{theorem}\label{thm:main}
For any $\epsilon > 0$, there is a $(1 + \ln 2 + \epsilon)$-approximation algorithm for WTAP.
\end{theorem}
In particular, for small enough $\epsilon >0$, we have $1 + \ln 2 + \epsilon < 1.7$. Our approach significantly deviates from the numerous recent techniques introduced in the context of TAP, i.e., the unweighted version of WTAP. More precisely, we develop a relative greedy algorithm, a method introduced by~\textcite{zelikovsky_1996_better} in the context of the Steiner Tree problem. 
\textcite{cohen_2013_approximation} later employed this approach to get a better-than-$2$ approximation for WTAP with bounded diameter, which inspired this work. Like other relative greedy algorithms, we iteratively contract well-chosen components. However, contrary to previous relative greedy approaches, including the above-mentioned ones, we rely on a (exponentially large) class of super-constant size components.
Nevertheless, we can efficiently find, within our novel class of components, the best one to contract next. This circumvents a key barrier of~\cite{cohen_2013_approximation}, namely that their components must have constant size to be able to enumerate over them, which requires the underlying tree to have constant diameter.
Moreover, we prove an approximate decomposition theorem for our components, which guarantees the existence of a good way to split any WTAP solution into candidate components for contraction.
This theorem, applied to an optimal WTAP solution, guarantees the existence of good components to contract.
We provide further details on our approach in Section~\ref{sec:relative_greedy}.

\subsection{Further related work}\label{sec:furtherRelatedWork}

Even though significant work on tree augmentation did not improve on the canonical approximation factor of $2$ for WTAP, remarkable progress has been achieved for numerous special cases. In particular, for the unweighted version, which is often simply called the \emph{Tree Augmentation Problem (TAP)}, a long line of research
~\cite{%
adjiashvili_2018_beating,%
cheriyan_2018_approximating_a,%
cheriyan_2018_approximating_b,%
cheriyan_2008_integrality,%
cohen_2013_approximation,%
even_2009_approximation,%
fiorini_2018_approximating,%
frederickson_1981_approximation,%
grandoni_2018_improved,%
khuller_1993_approximation,%
kortsarz_2016_simplified,%
kortsarz_2018_lp-relaxations,%
nagamochi_2003_approximation,%
nutov_2017_tree,%
grandoni_2018_improved%
}
led to the currently best approximation factor of $1.393$~\cite{cecchetto_2021_bridging}. This result even applies to the more general unweighted connectivity augmentation problem, which asks to increase the edge-connectivity of an arbitrary graph $G$ by one unit. (See also~\cite{byrka_2020_breaching,nutov_2020_approximation} for further recent results on connectivity augmentation.)
Starting with the work of \textcite{adjiashvili_2018_beating} and an elegant strengthening thereof by~\textcite{fiorini_2018_approximating}, approximation factors below~$2$ have been obtained for WTAP in the special case where the ratio of largest to smallest link weight is bounded by a constant. Subsequent further improvements on these procedures by~\textcite{grandoni_2018_improved} and~\textcite{cecchetto_2021_bridging} allowed for achieving approximation factors below $1.5$ for this case, and an elegant application of the round-or-cut framework, first employed in this context by \textcite{nutov_2017_tree}, allows for obtaining better-than-$2$ approximations for TAP even if the ratio of largest to smallest link weight is at most logarithmic in the size of the graph. Unfortunately, all these advances very deeply exploit that the max-to-min weight ratio of the links are bounded, and it seems highly unclear whether and how they could potentially be extended to WTAP.
Other special cases of WTAP, where an improvement over the approximation factor of $2$ has been achieved, is when the given tree has bounded diameter (see~\textcite{cohen_2013_approximation}), or when an optimal solution to a natural LP has no small fractional values~(see~\textcite{iglesias_2018_coloring}).

A natural generalization of WTAP is the \emph{$2$-edge-connected spanning subgraph problem} ($2$-ECSS). Contrary to WTAP, instead of starting with a spanning tree, one starts with an empty graph and needs to find a minimum cost set of edges leading to a $2$-edge-connected graph spanning all vertices. WTAP can be cast as $2$-ECSS by assigning to all tree edges a weight of zero such that they can be selected for free. Also for $2$-ECSS, the best known approximation factor is $2$, which can be achieved through a variety of elegant techniques~\cite{khuller_1994_biconnectivity}, including primal-dual methods~\cite{goemans_1995_general} and iterative rounding~\cite{jain_2001_factor} (see also~\cite{lau_2011_iterative,williamson_2011_design}). Progress beyond the factor $2$ has only been achieved for unweighted $2$-ECSS, where the task is to select a smallest number of edges to obtain a $2$-edge-connected spanning subgraph. After the first improvements~\cite{khuller_1994_biconnectivity,cheriyan_2001_improving}, this led to the currently best-known approximation factor of~$\sfrac{4}{3}$~\cite{hunkenschroder_2019_approximation,sebo_2014_shorter}.

\subsection{Organization of the paper}

We start with some very brief preliminaries in Section~\ref{sec:preliminaries}, which allows us to introduce basic terminology and notation, and formalize the well-known interpretation of WTAP as a covering problem.
In Section~\ref{sec:relative_greedy}, we then describe our relative greedy algorithm together with the main underlying results and show how they lead to Theorem~\ref{thm:main}. In particular, we also introduce our new class of components and describe their properties needed for our relative greedy procedure. These properties are proved in the following sections. More precisely, Section~\ref{sec:decomposition_proof} proves a decomposition theorem that guarantees that we can make progress by contracting a component that minimizes a natural selection function. Finally, in Section~\ref{sec:dynamic_program}, we show how a component minimizing the selection function can be found efficiently through dynamic programming.
\section{Preliminaries}\label{sec:preliminaries}

Let $(G=(V,E),L,w)$ be a WTAP instance. WTAP is naturally described as a covering problem, where the task is to select a minimum weight set of links that cover all $1$-cuts of the given tree $G$, which are the cuts containing a single tree edge $e$. Hence, the $1$-cuts correspond to the edges of $E$. The $1$-cuts that a single link $\ell\in L$ covers are described by the edge set $P_\ell\subseteq E$ of the unique path in $G$ between the endpoints of $\ell$. Because the addition of a link set $F\subseteq L$ to $G$ makes the graph $2$-edge-connected if and only if $F$ covers all $1$-cuts of $G$, WTAP can be formalized as the following natural covering problem.
\begin{equation}
\min\left\{  \sum_{\ell\in F} w(\ell) \colon F\subseteq L, \bigcup_{\ell\in F}P_\ell = E\right\}\tag{WTAP}
\end{equation}
We also say that a link $\ell$ \emph{covers} the edges $P_\ell$. Hence, a set of links is a WTAP solution if its links cover all edges.

For a link $\ell\in L$, we denote by $V_\ell\subseteq V$ the vertices of the path with edge set $P_\ell$, including its endpoints.
It is often convenient to assume that the WTAP instance is \emph{shadow-complete}, which means that for every link $\ell \in L$ and every two distinct vertices $v_1,v_2 \in V_\ell$, there is also a link $\{v_1,v_2\} \in L$ of same weight as $\ell$. In this case $\{v_1,v_2\}$ is called a \emph{shadow} of $\ell$. Note that a link $\ell_1\in L$ is a shadow of $\ell_2\in L$ if and only if $P_{\ell_1} \subseteq P_{\ell_2}$. Clearly, any WTAP instance can be transformed into an equivalent shadow-complete one by adding, for each link $\ell\in L$, all of its shadows, each with weight $w(\ell)$. Any such added shadow of a link $\ell\in L$ that gets selected in a solution can later be replaced by $\ell$ without changing the weight of the solution.

\section{Relative greedy algorithm for WTAP}\label{sec:relative_greedy}

As mentioned, our $(1+\ln 2 + \epsilon)$-approximation for WTAP is a relative greedy algorithm. The concept of relative greedy algorithms has been introduced by~\textcite{zelikovsky_1996_better} in the context of the Steiner Tree problem, and was later leveraged by~\textcite{cohen_2013_approximation} for WTAP with bounded diameter. The idea is to start with a weak but very well-structured approximation to the problem at hand, and then successively improve this solution by replacing parts of it. In the context of Steiner Tree,~\textcite{zelikovsky_1996_better} started with a simple $2$-approximation obtained by computing a minimum spanning tree over the terminals. This solution was then improved by successively finding an edge set to connect constantly many terminals, which is also called a \emph{component}, such that the cost of the component is (significantly) cheaper than the cost of the spanning tree edges that can be removed after including the component in the solution.

For WTAP, we start with the same highly structured $2$-approximation as~\textcite{cohen_2013_approximation} did, namely one only using so-called \emph{up-links} that are non-overlapping. Given a WTAP instance $(G=(V,E), L, w)$, we fix an arbitrary vertex $r\in V$, which we call \emph{root} from now on; an \emph{up-link} (with respect to $r$) is a link $\ell \in L$ such that one of its endpoints is on the unique path in $G$ between the root and the other endpoint. The statement below formalizes the properties of the $2$-approximate starting solution that our relative greedy procedure aims at improving. We denote by $L_{\mathrm{up}}\subseteq L$ the set of all up-links and by $\OPT\subseteq L$ an optimal solution to the WTAP instance.

\begin{lemma}[\cite{cohen_2013_approximation}]\label{lem:2_approximation}\footnote{%
This lemma follows from the observation that an up-link only solution that is at most a factor of $2$ heavier than $\OPT$ exists because one can start with $\OPT$ and replace each $\OPT$-link $\ell$ by two shadows $\ell_1,\ell_2$ of $\ell$ that are up-links and fulfill $P_{\ell} = P_{\ell_1} \cup P_{\ell_2}$.
Moreover, any up-link only solution can efficiently be transformed into one with only non-overlapping up-links that is no heavier by shortening up-links if necessary, i.e., replacing them by strict shadows. 
Finally, it remains to note that a cheapest non-overlapping up-link only solution can be found efficiently. This can be done through a dynamic program, or, alternatively, one can compute an optimal vertex solution to the canonical linear program, which is naturally integral because its constraint matrix is totally unimodular.}
Let $(G=(V,E),L,w)$ be a shadow-complete instance of WTAP.
Then we can in polynomial time compute a WTAP solution $U\subseteq L_{\mathrm{up}}$ such that
\begin{itemize}
\item $w(U) \le 2 \cdot w(\OPT)$, and
\item the edge sets $P_u$ for $u\in U$ are pairwise disjoint.
\end{itemize}
\end{lemma}

Starting with a $2$-approximate up-link solution $U^*\subseteq L_{\mathrm{up}}$ as guaranteed by Lemma~\ref{lem:2_approximation}, we seek to identify a subset of the links in $U^*$ that can be replaced by a cheaper link set $C\subseteq L$. To this end, for any set of up-links $U\subseteq L_{\mathrm{up}}$ and link set $C\subseteq L$, we denote by
\begin{equation*}
\Drop_U(C) \coloneqq \left\{u\in U\colon P_u \subseteq \bigcup_{\ell\in C} P_{\ell}\right\}
\end{equation*}
the links of $U$ that only cover a subset of the edges covered by links in $C$, and can thus safely be removed from a solution once $C$ is added.

In a general replacement step, we have some up-links $U\subseteq U^*$ with $U\neq \emptyset$ left in our current solution, and we seek to find a link set $C\in \mathfrak{L}$, where $\mathfrak{L} \subseteq 2^L$ is a well-chosen family, such that $C$ is a minimizer of 
\begin{equation}\label{eq:selectFuncOptimization}
\min\left\{\frac{w(F)}{w(\Drop_U(F))} \colon F\in \mathfrak{L}\right\}\enspace,
\end{equation}
i.e., it has the best ratio between the weight $w(C)$ of the links to be added versus the weight of the links in $\Drop_U(F)$, which can safely be removed. By convention, we interpret $\sfrac{w(F)}{w(\Drop_U(F))}$ as being $\infty$ whenever $w(\Drop_U(F))=0$, i.e., $\Drop_U(F)=\emptyset$.

The main challenge in such a relative greedy approach lies in finding a strong family $\mathfrak{L}\subseteq 2^L$ which simultaneously fulfills the following desired properties for any set $U\subseteq L_{\mathrm{up}}$ of non-overlapping up-links:
\begin{enumerate}[topsep=0.3em,label=(\alph*)]
\item\label{item:propFindBestComp} The minimization problem~\eqref{eq:selectFuncOptimization} can be solved efficiently.
\item\label{item:goodProgressWhenWUisLarge} If $w(U)$ is significantly heavier than $w(\OPT)$, then there is a set $C\in \mathfrak{L}$ for which $\sfrac{w(C)}{w(\Drop_U(C))}$ is significantly below $1$.
\end{enumerate}
Analogous to prior work on relative greedy algorithms, we call the sets in the family $\mathfrak{L}$ also \emph{components}. So far, relative greedy algorithms relied on constant-size components. In particular, for Steiner Tree, \citeauthor{zelikovsky_1996_better}'s~\cite{zelikovsky_1996_better} components were edge sets connecting constantly many terminals, and for WTAP with bounded diameter, \textcite{cohen_2013_approximation} considered sets of constantly many links.

Constant-size components have the obvious benefit that they allow for obtaining property~\ref{item:propFindBestComp} in a straightforward way. Moreover, their simple definition often makes it much easier to understand what one can achieve with such components, i.e., what improvements are possible by adding such a component to a solution. However, as we highlight in Figure~\ref{fig:constant_size_not_sufficient}, constant-size components do not allow in general for improving an up-link solution to a better-than-$2$ approximation. Hence, the bounded diameter restriction in~\cite{cohen_2013_approximation} is crucial for constant-size components to work.

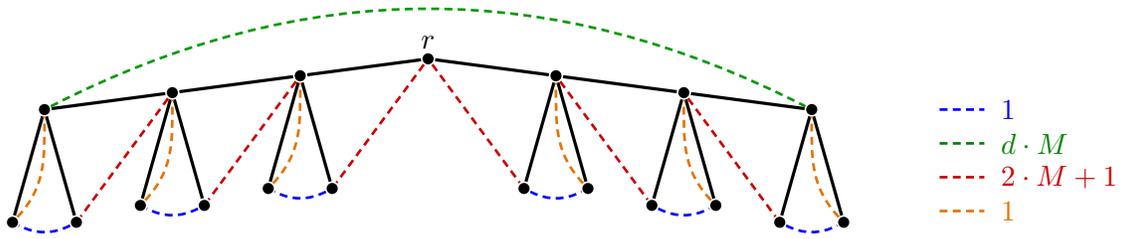
\begin{figure}[!ht]
\begin{center}
\begin{tikzpicture}[yscale=0.75,xscale=0.85]

\tikzset{
lks/.style={line width=1pt, densely dashed},
}

\begin{scope}[every node/.style={thick,draw=black,fill=black,circle,minimum size=0pt, inner sep=1.2pt, outer sep=1pt}]

\node (1) at (0,2) {};
\node (2) at (2,2.3) {};
\node (3) at (4,2.6) {};
\node (0) at (6,2.9) {};
\node (4) at (8,2.6) {};
\node (5) at (10,2.3) {};
\node (6) at (12,2) {};

\node (1a) at (0.5,0) {};
\node (2a) at (2.5,0.3) {};
\node (3a) at (4.5,0.6) {};
\node (4a) at (7.5,0.6) {};
\node (5a) at (9.5,0.3) {};
\node (6a) at (11.5,0) {};

\node (1b) at (-0.5,0) {};
\node (2b) at (1.5,0.3) {};
\node (3b) at (3.5,0.6) {};
\node (4b) at (8.5,0.6) {};
\node (5b) at (10.5,0.3) {};
\node (6b) at (12.5,0) {};

\end{scope}

\node[above] (r) at (0) {$r$};

\begin{scope}[very thick]
\draw (1) -- (2) -- (3) -- (0) -- (4) -- (5) -- (6);
\draw (1a) -- (1) -- (1b);
\draw (2a) -- (2) -- (2b);
\draw (3a) -- (3) -- (3b);
\draw (4a) -- (4) -- (4b);
\draw (5a) -- (5) -- (5b);
\draw (6a) -- (6) -- (6b);

\end{scope}

\begin{scope}[lks]

\draw[green!60!black] (1) to[out =30, in=150] (6);
\begin{scope}[darkred]
\draw (2) -- (1a);
\draw (3) -- (2a);
\draw (0) -- (3a);
\draw (0) -- (4a);
\draw (4) -- (5a);
\draw (5) -- (6a);
\end{scope}
\begin{scope}[orange!90!black]
\draw (1b) to[in=-90,out=50] (1);
\draw (2b) to[in=-90,out=50] (2);
\draw (3b) to[in=-90,out=50] (3);
\draw (4b) to[in=-90,out=130] (4);
\draw (5b) to[in=-90,out=130] (5);
\draw (6b) to[in=-90,out=130] (6);
\end{scope}
\begin{scope}[blue]
\draw (1a) to[bend left=30] (1b);
\draw (2a) to[bend left=30] (2b);
\draw (3a) to[bend left=30] (3b);
\draw (4a) to[bend right=30] (4b);
\draw (5a) to[bend right=30] (5b);
\draw (6a) to[bend right=30] (6b);
\end{scope}

\end{scope}

\begin{scope}[shift={(14,2)}]%
\def\ll{8mm} %
\def\vs{6mm} %

\begin{scope}[lks]
\draw[blue] (0,0) -- (0.7,0);
\draw[green!50!black] (0,-\vs) -- (0.7,-\vs);
\draw[darkred] (0,-2*\vs) -- (0.7,-2*\vs);
\draw[orange!90!black] (0,-3*\vs) -- (0.7,-3*\vs);
\end{scope}
\node[right,blue] at (\ll,0) {$1$};
\node[right,green!50!black] at (\ll,-\vs) {$d\cdot M$};
\node[right,darkred] at (\ll,-2*\vs) {$2\cdot M+1$};
\node[right,orange!90!black] at (\ll,-3*\vs) {$1$};
\end{scope}

\end{tikzpicture}
 \end{center}

\caption{An instance of WTAP showing that components of constant size are not sufficient to achieve an approximation ratio below $2$.
The tree $G$ is shown in black and dashed lines represent links. The link set $L$ consists of the drawn links together with all their shadows. The weights of the links are $1$ for the blue and orange links, $d\cdot M$ for the green link, and $2\cdot M +1$ for the red links. Here, $d$ is the number of red/blue/orange links, i.e., $d=6$ in this example, and $M$ is a  large constant.
\\
Then the union of the green link and the blue links is the unique optimal solution $\OPT$. The red and orange links form together a solution $U$ with $w(U)= 2\cdot w(\OPT)$. All red and orange links are up-links and the edge sets $P_u$ with $u\in U$ are pairwise disjoint.  ($U$ is even a cheapest up-link only solution, i.e., our 2-approximation algorithm might indeed output this solution.)
However, if we consider any set $C\subseteq L$ of at most $\sfrac{d}{2}$ links, then $w(\Drop_U(C))$ is at most $w(C)$.
This shows that, in general, we cannot improve the $2$-approximation $U$ by replacing a subset of $U$ by at most $k$ other links for some constant $k$.
}\label{fig:constant_size_not_sufficient}
\end{figure}

The components $\mathfrak{L}\subseteq 2^L$ we use are link sets with at most constant overlap on vertices. We call such link sets \emph{$O(1)$-thin}, as formalized below.
\begin{definition}[$k$-thin link set]
Let $(G=(V,E),L,w)$ be a WTAP instance and let $k\in \mathbb{Z}_{\geq 0}$. A link set $C\subseteq L$ is \emph{$k$-thin} if, for each $v\in V$, we have $|\{\ell\in C\colon v\in V_\ell\}|\leq k$.
\end{definition}
For our algorithm, we fix the constant $k$ in the above definition depending on the error $\epsilon > 0$ of our $(1+\ln 2 + \epsilon)$-approximation, by defining our components $\mathfrak{L}\subseteq 2^L$ to be all $\lceil\sfrac{2}{\epsilon}\rceil$-thin link sets.
Despite their super-constant size, this definition of components $\mathfrak{L}$ allows for efficiently solving~\eqref{eq:selectFuncOptimization} through a dynamic program, as formalized in the lemma below. Hence, our components $\mathfrak{L}$ fulfill property~\ref{item:propFindBestComp}.
\begin{restatable}{lemma}{LemmaOptimalComponents}\label{lem:finding_optimal_thin_components}
Let $k\in \mathbb{Z}_{\geq 0}$ be a constant, $(G=(V,E),L,w)$ be a WTAP instance, and $U\subseteq L_{\mathrm{up}}$ such that the edge sets $P_u$ for $u\in U$ are pairwise disjoint. Then we can compute in polynomial time a minimizer of 
\begin{equation*}
\min \left\{\frac{w(C)}{w(\Drop_U(C))} \colon C\subseteq L \text{ is $k$-thin}\right\}\enspace.
\end{equation*}
\end{restatable}

We prove Lemma~\ref{lem:finding_optimal_thin_components} in Section~\ref{sec:dynamic_program}.
Our relative greedy algorithm for WTAP is described in Algorithm~\ref{algo:relative_greedy}, which, due to Lemma~\ref{lem:finding_optimal_thin_components}, is a polynomial-time procedure. In the algorithm, and discussion later on, we assume that $\epsilon >0 $ is a fixed constant.

\begin{algorithm2e}[H]
\KwIn{A shadow-complete WTAP instance $(G=(V,E),L,w)$.}
\KwOut{A WTAP solution $F\subseteq L$ with $w(F) \le (1+\ln(2)+\epsilon)\cdot w(\OPT)$.}
\vspace*{2mm}
\begin{enumerate}[label=\arabic*.,ref=\arabic*,rightmargin=7mm]
\item\label{item:compute_U} Compute a WTAP solution $U\subseteq L_{\mathrm{up}}$ with $w(U) \le 2 \cdot w(\OPT)$ and disjoint $P_u$ for $u\in U$.
\item Initialize $F \coloneqq \emptyset$.
\item\label{item:choose_comp} While $U \neq \emptyset$:
\begin{itemize}[itemsep=0.5em]
\item Compute a minimizer $\displaystyle C\in \argmin\left\{\frac{w(C)}{w(\Drop_U(C))}\colon C\subseteq L \text{ is $\left\lceil\sfrac{2}{\epsilon}\right\rceil$-thin}\right\}$.

\item Add $C$ to $F$ and replace $U$ by $U\setminus \Drop_U(C)$.
\end{itemize}
\smallskip
\item Return $F$.
\end{enumerate}
\caption{Relative greedy algorithm for WTAP}\label{algo:relative_greedy}
\end{algorithm2e}

To make sure that our relative greedy algorithm is able to achieve approximation factors strictly below~$2$, it remains to show a formal version of property~\ref{item:goodProgressWhenWUisLarge}, i.e., that there is a profitable replacement step whenever $w(U)$ is significantly larger than $w(\OPT)$.

In order to prove that there is an improving replacement when we start, a natural reasoning, which has been used similarly in prior work, is as follows. Consider any up-link WTAP solution $U\subseteq L_{\mathrm{up}}$ with disjoint sets $P_u$ for $u\in U$. Ideally, we would like to find a partition of $\OPT$ into $\lceil\sfrac{2}{\epsilon}\rceil$-thin components $C_1,\ldots, C_q$ such that, for each $u\in U$, there is one component $C_j$ with $u\in \Drop_U(C_j)$. Notice that such a decomposition of $\OPT$ would immediately imply $\sum_{j=1}^q w(\Drop_U(C_j)) \geq w(U)$. Hence, by an averaging argument we have that if $w(U) > w(\OPT)$, then there exists a $\lceil\sfrac{2}{\epsilon}\rceil$-thin component $C_j$ with $w(C_j) < w(\Drop_U(C_j))$. As shown by~\textcite{cohen_2013_approximation}, this strategy works out in the bounded diameter case when dealing with constant-size components. However, such a decomposition does not exist in general WTAP instances, even when using the significantly more general class of $O(1)$-thin components.\footnote{Figure~\ref{fig:arbitrary_F_u} shows a bad example where such a decomposition does not exist. Indeed, to make sure that every link in $U$ is covered by at least one component, all links of the WTAP solution $F$ must be in the same component. As $v\in V_\ell$ for all $m+1$ links $\ell\in F$, this leads to a component that is not $m$-thin, where $m$ can be chosen arbitrarily large in the example.}
Nevertheless, as stated below, we can show that a slightly weaker statement holds, namely that such a decomposition exists if we first remove a well-chosen subset of up-links $R\subseteq U$ of small total weight. 
We later invoke the theorem with $U$ being an up-link WTAP solution as guaranteed by Lemma~\ref{lem:2_approximation} and the WTAP solution $F$ being $\OPT$.

\begin{restatable}[decomposition theorem]{theorem}{decompositionTheorem}\label{thm:decomposition}
Let $(G=(V,E),L,w)$ be a WTAP instance, $F\subseteq L$ be a WTAP solution, and let $U\subseteq L_{\mathrm{up}}$ be a set of up-links such that the sets $P_u$ with $u\in U$ are pairwise disjoint.
Then, for any $\epsilon >0$, there exists a partition $\Cscr$ of $F$ into $\lceil\sfrac{1}{\epsilon}\rceil$-thin sets and a set $R\subseteq U$ such that
\begin{enumerate}
\item for every $u\in U\setminus R$, there exists some $C \in \Cscr$ such that $P_u \subseteq \bigcup_{\ell\in C} P_{\ell}$, and  \label{item:links_outside_R_covered}
\item $w(R) \le \epsilon \cdot w(U)$.
\end{enumerate}
\end{restatable}

We prove the decomposition theorem in Section~\ref{sec:decomposition_proof}.
Using Theorem~\ref{thm:decomposition}, we readily obtain that Algorithm~\ref{algo:relative_greedy} has the desired approximation guarantee by leveraging known arguments (see, e.g.,~\cite{zelikovsky_1996_better,gropl_2001_approximation,cohen_2013_approximation}). For completeness, we provide a self-contained proof below.

\begin{theorem}
For every $\epsilon >0$, Algorithm~\ref{algo:relative_greedy} is a $(1+\ln 2 +\epsilon)$-approximation algorithm for WTAP.
\end{theorem}
\begin{proof}
Throughout the algorithm we maintain the invariant that $U\cup F$ is a WTAP solution. Hence, the returned link set $F$ is indeed a WTAP solution and it remains to bound its weight.

Let $U_0$ be the link set $U$ computed in step~\ref{item:compute_U} of Algorithm~\ref{algo:relative_greedy} and let $U_i$ denote the set $U$ at the end of the $i$-th iteration of the while loop in step~\ref{item:choose_comp}. 
Let $C_i$ denote the component $C$ chosen in the $i$-th iteration.
We apply Theorem~\ref{thm:decomposition} to an optimal WTAP solution $\OPT$ and the link set $U_0$ to obtain a set $R\subseteq U_0$ with $w(R)\le \frac{1}{2}\epsilon \cdot w(U_0) \le \epsilon \cdot w(\OPT)$ and a partition $\Cscr$ of $\OPT$ into $\lceil\sfrac{2}{\epsilon}\rceil$-thin sets.

Consider the $i$-th iteration of the while loop.
Because there is a component $C\in \Cscr$ with $P_u \subseteq \bigcup_{\ell\in C} P_{\ell}$ for every $u\in U_{i-1}\!\setminus R$, we have
\begin{equation*}
\sum_{C\in \Cscr} w\bigl(\Drop_{U_{i-1}}(C)\bigr)\ \ge\ w(U_{i-1}\!\setminus R)\ \ge\ w(U_{i-1}) - w(R)\enspace.
\end{equation*}
This implies 
\begin{equation*}
\min_{C\in \Cscr} \frac{w(C)}{w(\Drop_{U_{i-1}}(C))}
\ \le\  
\frac{\sum_{C\in \Cscr} w(C)}{\sum_{C\in \Cscr} w(\Drop_{U_{i-1}}(C))}
\ \le\ 
\frac{w(\OPT)}{w(U_{i-1}) - w(R)}\enspace.
\end{equation*}
Because every $C\in \Cscr$ is $\lceil\sfrac{2}{\epsilon}\rceil$-thin, every component $C\in \Cscr$ could have been chosen by the algorithm in step~\ref{item:choose_comp}. Thus,
\begin{equation}\label{eq:upper_bound_ratio}
 \frac{w(C_i)}{w(U_{i-1}\setminus U_{i})}\ =\ \frac{w(C_i)}{w(\Drop_{U_{i-1}}(C_i))} \ \le\ \frac{w(\OPT)}{w(U_{i-1}) - w(R)} \enspace.
\end{equation}
Moreover, the algorithm could also have chosen any component consisting of a single link $u\in U_{i-1}$, which implies $\sfrac{w(C_i)}{w(U_{i-1}\setminus U_{i})} \le 1$. Combining this with \eqref{eq:upper_bound_ratio}, we obtain
\begin{equation*}
w(C_i) \ \le\ \min\left\{\frac{w(\OPT)}{w(U_{i-1}) -w(R)}, 1\right\} \cdot w(U_{i-1} \setminus U_{i})
\ \le\ \int_{w(U_i)}^{w(U_{i-1})} \min\left\{\frac{w(\OPT)}{x- w(R)},1\right\}\,dx\enspace,
\end{equation*}
where we used $w(U_i) \le w(U_{i-1})$ and that $\min\{\frac{w(\OPT)}{x-w(R)},1\}$ is monotonically decreasing in $x$.

Let $m$ denote the number of iterations of the while loop.
Then $U_m=\emptyset$ and 
\begin{align*}
w(F)\ =\ \sum_{i=1}^m w(C_i) \ 
\le&\ \sum_{i=1}^m  \int_{w(U_i)}^{w(U_{i-1})} \min\left\{\frac{w(\OPT)}{x- w(R)},1\right\}\,dx\\[2mm]
=&\ \int_{w(U_m)}^{w(U_0)} \min\left\{\frac{w(\OPT)}{x-w(R)},1\right\}\,dx  \\[2mm]
=&\ \int_0^{w(\OPT)+w(R)} 1 \,dx + \int_{w(\OPT)+w(R)}^{w(U_0)}\frac{w(\OPT)}{x-w(R)} \,dx \\[2mm]
=&\ w(\OPT)+w(R) + \ln\left(\frac{w(U_0)-w(R)}{w(\OPT)}\right) \cdot w(\OPT) \\[2mm]
\le&\ (1 +\epsilon)\cdot w(\OPT) + \ln 2 \cdot w(\OPT)\enspace,
\end{align*}
where the last inequality follows from $w(R) \le \epsilon \cdot w(\OPT)$ and $w(U_0) \le 2 \cdot w(\OPT)$.
\end{proof}

As is common with relative greedy procedures, one does not need to run the while loop of Algorithm~\ref{algo:relative_greedy} until $U\neq \emptyset$, but can stop early. More precisely, as soon as a component $C\in \argmin\{\sfrac{w(C)}{w(\Drop_U(C))} \colon C\subseteq L \text{ is $\lceil\sfrac{\epsilon}{2}\rceil$-thin}\}$ is computed that does not improve the solution anymore, i.e., $\sfrac{w(C)}{w(\Drop_U(C))} \geq 1$, then one can return $F\cup U$ without continuing the while loop. (We recall that $\sfrac{w(C)}{w(\Drop_U(C))}\geq 1$ actually implies $\sfrac{w(C)}{w(\Drop_U(C))}= 1$ because $C$ can always be chosen to be a single up-link in $U$.) Indeed, once such a set $C$ is encountered, any future replacements of up-links by a component done in the while loop will also not improve the solution.

\section{Proving the decomposition theorem}\label{sec:decomposition_proof}

In this section we prove our decomposition theorem, Theorem~\ref{thm:decomposition}.
Hence, we are given a WTAP solution $F\subseteq L$ and a set $U\subseteq L_{\mathrm{up}}$ of up-links such that the edge sets $P_{u}$ for $u\in U$ are pairwise disjoint.
For every link $u\in U$, we will first fix a set $F_u \subseteq F$ satisfying $P_u \subseteq \bigcup_{\ell \in F_u} P_{\ell}$.
We will then choose a set $R\subseteq U$ with $w(R)\le \epsilon \cdot w(U)$ and a partition $\Cscr$ of $F$ into $\lceil\sfrac{1}{\epsilon}\rceil$-thin sets such that, for every $u\in U\setminus R$, there is a component $C\in\Cscr$ with $F_u \subseteq C$.
We emphasize that, to be able to achieve a decomposition as claimed by Theorem~\ref{thm:decomposition}, it is crucial to choose the set $F_u$ for $u\in U$ carefully.
In particular, a natural choice, used in a similar setting by~\textcite{cohen_2013_approximation}, would be to let $F_u$ be any minimal (or any minimum cardinality) set with $P_u \subseteq \bigcup_{\ell \in F_u} P_{\ell}$. However, the example highlighted in Figure~\ref{fig:arbitrary_F_u} shows that this renders it impossible to achieve the decomposition property with the strategy  outlined above. Before expanding on how we choose the sets $F_u$, we continue the overview of our proof strategy for the decomposition theorem.

\begin{figure}[!ht]
\begin{center}
\begin{tikzpicture}[scale=0.85]

\tikzset{
fs/.style={line width=1pt, blue, densely dashed},
us/.style={line width=1pt, darkred, densely dashed}
}

\clip (-0.5,-1.5) rectangle (14.5,6.5);

\begin{scope}[every node/.style={thick,draw=black,fill=black,circle,minimum size=0pt, inner sep=1.2pt, outer sep=1pt}]

\node (1) at (6,6) {};
\node (1a) at (7,5) {};
\node (2) at (5,5) {};
\node (2a) at (6,4) {};
\node (4) at (3,3) {};
\node (4a) at (4,2) {};
\node (5) at (2,2) {};
\node (5a) at (3,1) {};
\node (w) at (1,1) {};
\node (1b) at (0,0) {};
\node (2b) at (0.5,0) {};
\node (4b) at (1.5,0) {};
\node (5b) at (2,0) {};
\end{scope}

\node[above=1pt] (r) at (1) {$r$};
\node[above left] (wlabel) at (w) {$v$};
\node[rotate=45] (dots1) at (4,4) {$\dots$};
\node (dots2) at (1,0) {$\dots$};

\begin{scope}[very thick]
\draw (1)--(1a);
\draw (2)--(2a);
\draw (4)--(4a);
\draw (5)--(5a);

\draw (1) --(2);
\draw (2) -- (4.3,4.3);
\draw (3.6,3.6) --(4);
\draw (4) --(5) --(w);

\draw (w) -- (1b);
\draw (w) -- (2b);
\draw (w) -- (4b);
\draw (w) -- (5b);
\end{scope}

\begin{scope}[us]
\draw (1) --(2a);
\draw (4,3.4) -- (4a);
\draw (4) -- (5a);
\end{scope}

\begin{scope}[darkred]
\node[right] (u1) at (5.95,5) {$u_1$};
\node[right] (u4) at (3.95,3) {$u_{m-1}$};
\node[right] (u5) at (2.95,2) {$u_m$};
\end{scope}

\begin{scope}[fs]

\draw (5b) to[out=30, in=-125] (5a);
\draw(4b) to[out=-70, in=-85]  (4a);
\draw (2b) to[out=-60, in=-75] (2a);
\draw (1b) to[out=-70, in=-65](1a);
\end{scope}

\begin{scope}[blue]
\node[right] (l1) at (7.2,3.8) {$\ell_0$};
\node[right] (l2) at (6.1,3.3) {$\ell_1$};
\node[right] (l4) at (3.9,1.5) {$\ell_{m-1}$};
\node[above left] (l5) at (2.7,0.4) {$\ell_m$};
\end{scope}

\begin{scope}[shift={(10,3)}]%
\def\ll{30mm} %
\def\vs{8mm} %

\node[right] at (0,0) {$G=(V,E)$};
\node[right,blue] at (0,-\vs) {$F=\{\ell_0,\ell_1,,\dots,\ell_m\}$};
\node[right,darkred] at (0,-2*\vs) {$U=\{u_1,\dots,u_m\}$};
\end{scope}

\end{tikzpicture}
 \end{center}
\caption{The figure shows an example with a tree $G=(V,E)$ (black), a WTAP solution $F$ (blue), and a set $U$ of up-links (red) such that the edge sets $P_u$ with $u\in U$ are pairwise disjoint.
In this example, $v\in V_\ell$ for every link $\ell \in F$.
Therefore, any $k$-thin subset of $F$ contains at most $k$ links.
For every link $u_i \in U$, the set $F_u=\{\ell_0,\ell_i\}\subseteq F$ is a  minimal subset of $F$ with $P_{u_i} \subseteq \bigcup_{\ell\in F_u} P_{\ell}$. However, with this choice of $F_u$, every set $F_u$ with $u\in U$ contains the link $\ell_0$.
Thus, if we consider any partition $\Cscr$ of $F$ into $k$-thin components, the only component $C\in\Cscr$ for which we can have $F_u\subseteq C$ for some $u\in U$, is the component $C_0$ containing $\ell_0$.
Because this component is $k$-thin, it contains at most $k$ links and hence the number of links $u\in U$ with $F_u \subseteq C_0$ is at most $k-1$.
For $w(u)=1$ for all $u\in U$, this shows that the  total weight of the up-links $u\in U$ for which $F_u$ is not contained in any component $C\in \Cscr$ is at least $(1-\tfrac{k-1}{m})\cdot w(U)$ instead of at most $\epsilon \cdot w(U)$ as required.
}\label{fig:arbitrary_F_u}
\end{figure}
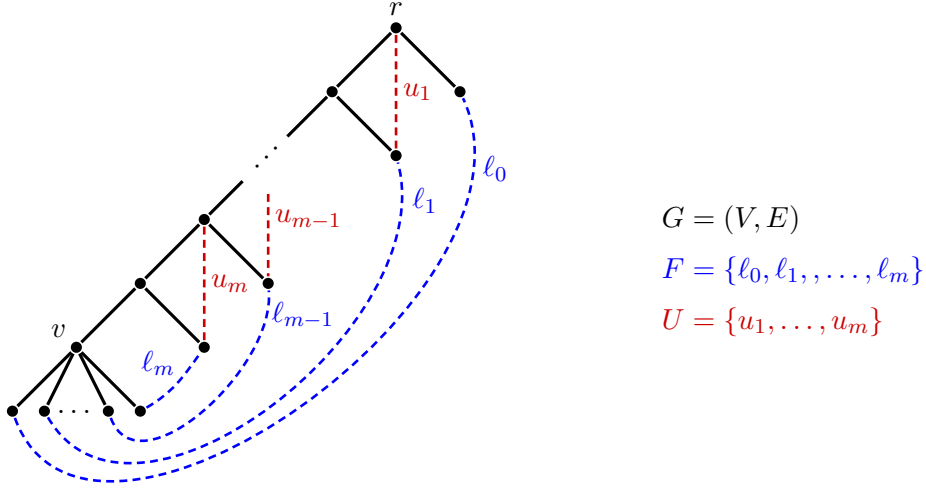

Once we fixed the sets $F_u$ for all $u\in U$ and the set $R\subseteq U$, the definition of the partition $\Cscr$ is straightforward.
Two links $\ell_1,\ell_2\in F$ are in the same component $C\in \Cscr$ if and only if there is an up-link $u\in U\setminus R$ with $\ell_1,\ell_2\in F_u$.
We express this dependency through a directed graph with vertex set $F$.
For every set $F_u$ with $u\in U\setminus R$, this graph contains a path $(F_u,A_u)$, which we formally define later.
Hence, the connected components of this dependency graph with vertex set $F$ yield the partition $\Cscr$.
More precisely, $C\subseteq F$ is a part of the partition $\Cscr$ if and only if the dependency graph has a connected component with vertex set $C$.
We remark that this construction of a dependency graph has been used before by~\textcite{cohen_2013_approximation} in the bounded diameter case with a different choice for the sets $F_u$.

To prove the decomposition theorem, we use such a dependency graph not only to find the partition $\Cscr$, but also to choose the set $R\subseteq U$.
We thus consider a dependency graph that has vertex set $F$ and contains arcs $A_u$ for every $u\in U$, where $A_u$ is again the above-mentioned arc set forming a path with vertex set $F_u$.
Then we show that there exists a set $R\subseteq U$ with $w(R)\le \epsilon \cdot w(U)$ such that, after removing the arcs in $\bigcup_{u\in R} A_u$ from the dependency graph, every connected component $(C,A)$ of the dependency graph fulfills that $C$ is $\lceil \sfrac{1}{\epsilon}\rceil$-thin.
To prove the existence of such a set $R$, we exploit that (with our choice of $(F_u,A_u)$ for $u\in U$) the dependency graph has the following two properties.
\begin{enumerate}[label=(\arabic*),topsep=3pt,itemsep=1pt]
\item\label{item:branching} The dependency graph for $U$ is a branching.
\item\label{item:thinness} Let $(C,A)$ be a connected component of the dependency graph. If the arc set of every directed path in $(C,A)$ has nonempty intersection with $A_u$ for at most $k$ up-links $u\in U$, then $C$ is $(k+1)$-thin.
\end{enumerate}
To obtain the first of these properties, we use that the sets $P_u$ for $u\in U$ are pairwise disjoint. Property~\ref{item:branching} was already shown to hold for the dependency graph used by \textcite{cohen_2013_approximation} and holds as long as we choose $F_u$ to be a minimal set with $P_u \subseteq \bigcup_{\ell \in F_u} P_{\ell}$.
To prove the second property, we crucially need our particular choice of $F_u$ as shown in Figure~\ref{fig:choice_F_u_crucial_dependency_graph}.

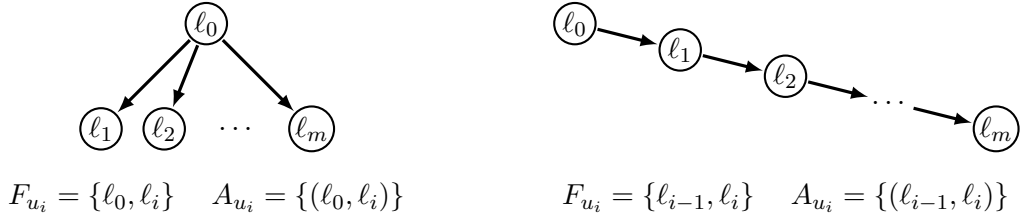
\begin{figure}[!ht]
\begin{center}
\begin{tikzpicture}[scale=0.7]

\tikzset{
fs/.style={line width=1pt, blue, densely dashed},
us/.style={line width=1pt, darkred, densely dashed}
}

\node () at (0,-1.3) {$F_{u_i} = \{\ell_0,\ell_i\}$ \quad $A_{u_i}=\{(\ell_{0},\ell_i)\}$};

\begin{scope}[every node/.style={thick,draw=black,fill=none,circle,minimum size=5.5mm, inner sep=0.5pt, outer sep=1pt}]
\node (l0) at (0,2) {$\ell_0$};
\node (l1) at (-2,0) {$\ell_1$};
\node (l2) at (-0.8,0) {$\ell_2$};
\node (lm) at (2,0) {$\ell_m$};
\end{scope}

\node (dots) at (0.6,0) {$\dots$};

\begin{scope}[very thick, ->, >=latex]
\draw (l0) to (l1);
\draw (l0) to (l2);
\draw (l0) to (lm);
\end{scope}

\begin{scope}[shift={(11,0)}]

\node () at (0,-1.3) {$F_{u_i} = \{\ell_{i-1},\ell_i\}$ \quad $A_{u_i}=\{(\ell_{i-1},\ell_i)\}$};

\begin{scope}[every node/.style={thick,draw=black,fill=none,circle,minimum size=5.5mm, inner sep=0.5pt, outer sep=1pt}]
\node (l0) at (-4,2) {$\ell_0$};
\node (l1) at (-2,1.5) {$\ell_1$};
\node (l2) at (0,1) {$\ell_2$};
\node[draw=white] (l3) at (2,0.5) {$\dots$};
\node (lm) at (4,0) {$\ell_m$};
\end{scope}

\begin{scope}[very thick, ->, >=latex]
\draw (l0) to (l1);
\draw (l1) to (l2);
\draw (l2) to (l3);
\draw (l3) to (lm);
\end{scope}
\end{scope}

\end{tikzpicture}
 \end{center}
\caption{The figure shows the dependency graph resulting from different choices of $F_u$ for the example instance from Figure~\ref{fig:arbitrary_F_u}.
The left picture shows the dependency graph that we would get if we chose $F_{u_i}=\{\ell_0,\ell_i\}$ for $i\in \{1,\ldots, m\}$. We recall that we already argued in the caption of Figure~\ref{fig:arbitrary_F_u} that this choice makes it impossible to obtain the decomposition theorem as suggested. Also, one can see that property~\ref{item:thinness} is clearly not fulfilled for the left-hand side choice. Indeed, every directed path intersects with at most one set $A_u$, which, if property~\ref{item:thinness} were fulfilled, should imply $2$-thinness of the component; however, the thinness is $m+1$.
The right picture shows the dependency graph for the choice of $F_{u_i}$ and $A_{u_i}$  that we will make to prove the decomposition theorem. In contrast to the left picture, the choice in the right picture fulfills property~\ref{item:thinness}: The set $F=\{\ell_0, \ell_1, \dots, \ell_m\}$ is $(m+1)$-thin and the dependency graph is a path that has nonempty intersection with all $m$ sets $A_{u_i}$.
}
\label{fig:choice_F_u_crucial_dependency_graph}
\end{figure}

Once properties~\ref{item:branching} and~\ref{item:thinness} are shown, there is a simple choice of $R\subseteq U$ to obtain the decomposition theorem, as explained in Figure~\ref{fig:example_choice_R}. 

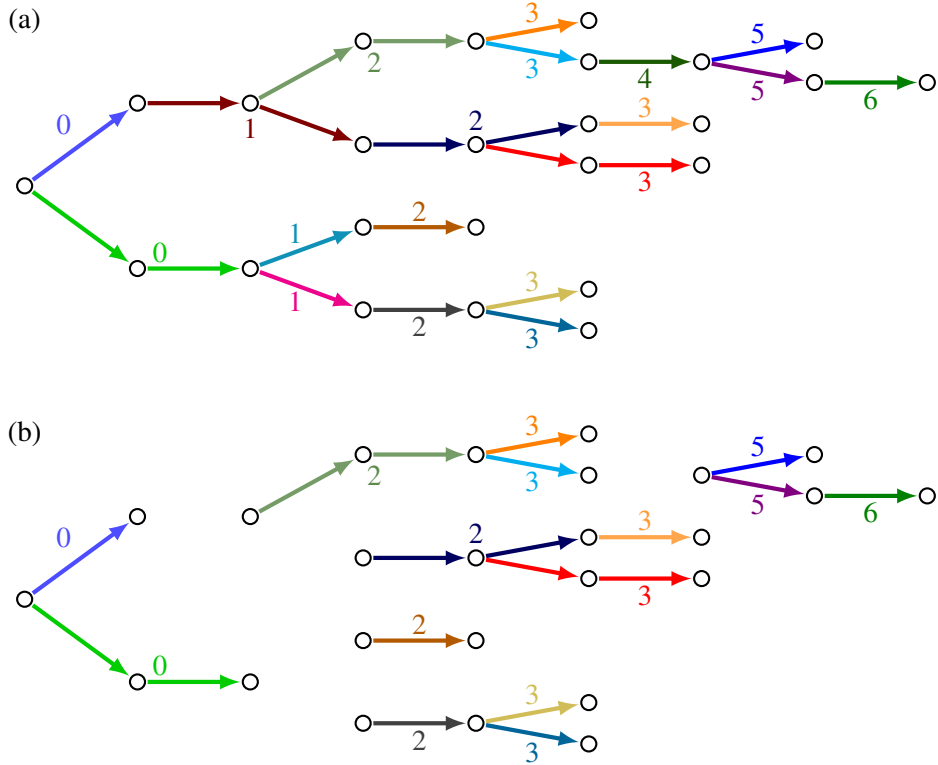
\begin{figure}[H]
\begin{center}
\begin{tikzpicture}[xscale=1.5, yscale=0.55]

\node (a) at (0,6) {(a)};

\begin{scope}[every node/.style={thick,draw=black,fill=none,circle,minimum size=2mm, inner sep=0.5pt, outer sep=1pt}]

\node (v1) at (0,2) {};
\node (v2) at (1,4) {};
\node (v3) at (2,4) {};
\node (v4) at (3,5.5) {};
\node (v5) at (4,5.5) {};
\node (v6) at (5,6) {};
\node (v7) at (5,5) {};
\node (v8) at (6,5) {};
\node (v9) at (7,5.5) {};
\node (v10) at (7,4.5) {};
\node (v11) at (8,4.5) {};
\node (v12) at (3,3) {};
\node (v13) at (4,3) {};
\node (v14) at (5,3.5) {};
\node (v15) at (6,3.5) {};
\node (v16) at (5,2.5) {};
\node (v17) at (6,2.5) {};
\node (v18) at (1,0) {};
\node (v19) at (2,0) {};
\node (v20) at (3,1) {};
\node (v21) at (4,1) {};
\node (v22) at (3,-1) {};
\node (v23) at (4,-1) {};
\node (v24) at (5,-0.5) {};
\node (v25) at (5,-1.5) {};
\end{scope}

\begin{scope}[ultra thick, ->, >=latex]
\draw[blue!70!white] (v1) to (v2);
\draw[red!50!black] (v2) to (v3);
\draw[red!50!black] (v3) to (v12);
\draw[darkgreen!60!white] (v3) to (v4);
\draw[darkgreen!60!white] (v4) to (v5);
\draw[orange](v5) to (v6);
\draw[cyan] (v5) to (v7);
\draw[darkgreen] (v7) to (v8);
\draw[blue] (v8) to (v9);
\draw[violet] (v8) to (v10);
\draw[green!50!black] (v10) to (v11);
\draw[darkblue] (v12) to (v13);
\draw[darkblue] (v13) to (v14);
\draw[orange!70!white] (v14) to (v15);
\draw[red] (v13) to (v16);
\draw[red] (v16) to (v17);
\draw[green!80!black] (v1) to (v18);
\draw[green!80!black] (v18) to (v19);
\draw[cyan!70!black] (v19) to (v20);
\draw[orange!70!black] (v20) to (v21);
\draw[magenta] (v19) to (v22);
\draw[gray!50!black] (v22) to (v23);
\draw[yellow!80!blue] (v23) to (v24);
\draw[blue!60!green](v23) to (v25);
\end{scope}

\begin{scope}
\node[above left, blue!70!white] (c1) at (0.5,3) {0};
\node[red!50!black] (c2) at (2,3.4) {1};
\node[darkgreen!70!white] (c3) at (3.1,5) {2};
\node[orange](c4) at (4.5,6.2) {3};
\node[cyan] (c5) at (4.5,4.85) {3};
\node[darkgreen] (c6) at (5.5,4.6) {4};
\node[blue] (c7) at (6.5,5.7) {5};
\node[violet] (c8) at (6.5,4.3) {5};
\node[green!50!black] (c9) at (7.5,4.1) {6};
\node[darkblue] (c10) at (4,3.6) {2};
\node[orange!80!white] (c11) at (5.5,3.9) {3};
\node[red] (c12) at (5.5,2.1) {3};
\node[green!80!black] (c13) at (1.2,0.4) {0};
\node[cyan!70!black] (c14) at (2.4,0.85) {1};
\node[orange!70!black] (c15) at (3.5,1.4) {2};
\node[magenta] (c16) at (2.4,-0.8) {1};
\node[gray!50!black] (c17) at (3.5,-1.4) {2};
\node[yellow!80!blue] (c18) at (4.5,-0.3) {3};
\node[blue!60!green] (c19) at (4.5,-1.7) {3};
\end{scope}

\begin{scope}[shift={(0,-10)}]

\node (b) at (0,6) {(b)};

\begin{scope}[every node/.style={thick,draw=black,fill=none,circle,minimum size=2mm, inner sep=0.5pt, outer sep=1pt}]

\node (v1) at (0,2) {};
\node (v2) at (1,4) {};
\node (v3) at (2,4) {};
\node (v4) at (3,5.5) {};
\node (v5) at (4,5.5) {};
\node (v6) at (5,6) {};
\node (v7) at (5,5) {};
\node (v8) at (6,5) {};
\node (v9) at (7,5.5) {};
\node (v10) at (7,4.5) {};
\node (v11) at (8,4.5) {};
\node (v12) at (3,3) {};
\node (v13) at (4,3) {};
\node (v14) at (5,3.5) {};
\node (v15) at (6,3.5) {};
\node (v16) at (5,2.5) {};
\node (v17) at (6,2.5) {};
\node (v18) at (1,0) {};
\node (v19) at (2,0) {};
\node (v20) at (3,1) {};
\node (v21) at (4,1) {};
\node (v22) at (3,-1) {};
\node (v23) at (4,-1) {};
\node (v24) at (5,-0.5) {};
\node (v25) at (5,-1.5) {};
\end{scope}

\begin{scope}[ultra thick, ->, >=latex]
\draw[blue!70!white] (v1) to (v2);
\draw[darkgreen!60!white] (v3) to (v4);
\draw[darkgreen!60!white] (v4) to (v5);
\draw[orange](v5) to (v6);
\draw[cyan] (v5) to (v7);
\draw[blue] (v8) to (v9);
\draw[violet] (v8) to (v10);
\draw[green!50!black] (v10) to (v11);
\draw[darkblue] (v12) to (v13);
\draw[darkblue] (v13) to (v14);
\draw[orange!70!white] (v14) to (v15);
\draw[red] (v13) to (v16);
\draw[red] (v16) to (v17);
\draw[green!80!black] (v1) to (v18);
\draw[green!80!black] (v18) to (v19);
\draw[orange!70!black] (v20) to (v21);
\draw[gray!50!black] (v22) to (v23);
\draw[yellow!80!blue] (v23) to (v24);
\draw[blue!60!green](v23) to (v25);
\end{scope}

\begin{scope}
\node[above left, blue!70!white] (c1) at (0.5,3) {0};
\node[darkgreen!70!white] (c3) at (3.1,5) {2};
\node[orange](c4) at (4.5,6.2) {3};
\node[cyan] (c5) at (4.5,4.85) {3};
\node[blue] (c7) at (6.5,5.7) {5};
\node[violet] (c8) at (6.5,4.3) {5};
\node[green!50!black] (c9) at (7.5,4.1) {6};
\node[darkblue] (c10) at (4,3.6) {2};
\node[orange!80!white] (c11) at (5.5,3.9) {3};
\node[red] (c12) at (5.5,2.1) {3};
\node[green!80!black] (c13) at (1.2,0.4) {0};
\node[orange!70!black] (c15) at (3.5,1.4) {2};
\node[gray!50!black] (c17) at (3.5,-1.4) {2};
\node[yellow!80!blue] (c18) at (4.5,-0.3) {3};
\node[blue!60!green] (c19) at (4.5,-1.7) {3};
\end{scope}
\end{scope}

\end{tikzpicture}
 \end{center}

\caption{Picture (a) shows an example of a connected component of the dependency graph for $U$. The arc set of every path $(F_u,A_u)$ is shown in a different color.
In order to construct the set $R\subseteq U$, we assign a label from $\mathbb{Z}_{\geq 0}$ to each of the sets $A_u$ that is equal to the number of different colors (or equivalently sets $A_{\overline{u}}$ for $\overline{u}\in U$) encountered on the unique path from the root of the component to the start of the path $A_u$.
The numbers in the figure show such a labeling. Let $k=\lceil \sfrac{1}{\epsilon} \rceil$.
The choice of the labeling guarantees that, for every $i\in\{0,\dots,k-1\}$, removing all arcs with any label $j\in \mathbb{Z}_{\geq 0}$ that satisfies $j\equiv i \pmod{k}$, leads to a dependency graph where at most $k-1$ different colors appear on each path. In other words, after removing the arcs with label $i \mod*{k}$, every path has nonempty intersection with at most $k-1$ different sets $A_u$. Then, by property~\ref{item:thinness}, the connected components correspond to $k$-thin sets.
Picture (b) shows an example for $i=1$ and $k=3$.
We choose $R$ to be the set of all up-links $u\in U$ for which $A_u$ has label $i \mod*{k}$, where $i\in\{0,\dots,k-1\}$ is an index for which the resulting set $R$ has minimum weight $w(R)$.
}\label{fig:example_choice_R}
\end{figure}

\subsection{The dependency graph}

Next, we formally define the dependency graph of a set $U\subseteq L_{\mathrm{up}}$.
Let $(G=(V,E),L,w)$ be a WTAP instance and let $F\subseteq L$ be a WTAP solution. Let $r\in V$ be the (arbitrarily chosen) root of $G$. The root defines a natural ancestry relationship. The \emph{ancestors} of a vertex $v\in V$ are all vertices $z\in V$ that lie on the unique $r$-$v$ path, which includes $r$ and $v$ (we talk about a \emph{strict ancestor} to disallow $z=v$). Analogously, $v\in V$ is a (strict) descendant of $z\in V$ if $z$ is a (strict) ancestor of $v$. For a link $\ell\in L$, we denote by $\apex(\ell)\in V$ the lowest common ancestor of the two endpoints of $\ell$, i.e., the vertex in $V_\ell$ closest to the root.

For each up-link $u\in L_{\mathrm{up}}$, we choose a minimal link set $F_u \subseteq F$ with $P_u \subseteq \bigcup_{\ell \in F_u} P_{\ell}$ as follows.
Let $u=\{t,b\}$ where $t$ is an ancestor of $b$.
We define $v_u$ to be the lowest ancestor of $t$, i.e., the ancestor farthest away from the root $r$, such that $P_u$ is covered by links in
\begin{equation*}
B_{v_u} \coloneqq \left\{ \ell\in F \colon \apex(\ell) \text{ is a descendant of } v_u \right\}\enspace,
\end{equation*}
i.e., $P_u \subseteq \bigcup_{\ell\in B_{v_u}} P_{\ell}$.  Then we choose $F_u \subseteq B_{v_u}$ minimal such that $P_u \subseteq \bigcup_{\ell \in F_u} P_{\ell}$. See Figure~\ref{fig:order_on_F_u} for an example of a minimal set $F_u$.

Due to minimality of $F_u$, the links $\ell\in F_u$ have a natural order. More precisely, this order is induced by how close to the root the edges of $P_{u,\ell} := P_{u} \setminus \bigcup_{\bar \ell \in F_u \setminus \{\ell\}} P_{\bar \ell}$\, are, which are the edges in $P_u$ for which $\ell$ is the only link in $F_u$ that covers them. (See Figure~\ref{fig:order_on_F_u}.)
To define this order formally, we start by observing that the sets $P_{u,\ell}$ are edge sets of paths.
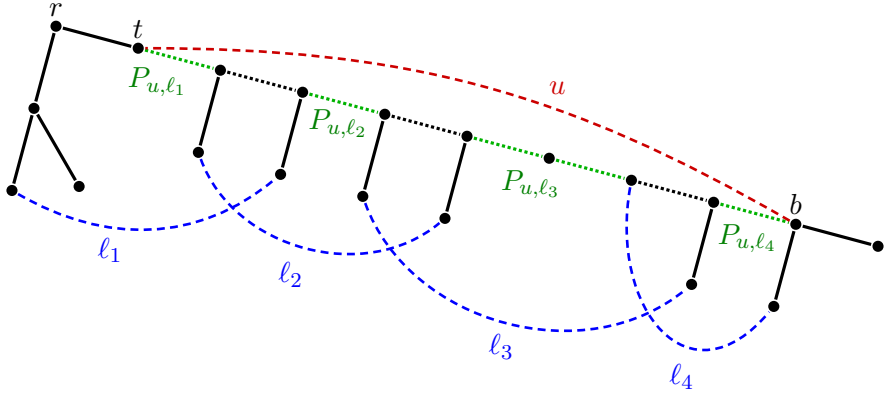
\begin{figure}[!ht]
\begin{center}
\begin{tikzpicture}[scale=0.8]

\tikzset{
fs/.style={line width=1pt, blue, densely dashed},
us/.style={line width=1pt, darkred, densely dashed}
}

\begin{scope}[every node/.style={thick,draw=black,fill=black,circle,minimum size=0pt, inner sep=1.2pt, outer sep=1pt}, rotate=30, xscale=-1]

\node (1) at (11,11) {};
\node (1a) at (12,10) {};
\node (1x) at (13,9) {};
\node (1y) at (12,8.5) {};
\node (2) at (10,10) {};
\node (3) at (9,9) {};
\node (3a) at (10,8) {};
\node (4) at (8,8) {};
\node (4a) at (9,7) {};
\node (5) at (7,7) {};
\node (5a) at (8,6) {};
\node (6) at (6,6) {};
\node (6a) at (7,5) {};
\node (7) at (5,5) {};
\node (8) at (4,4) {};
\node (9) at (3,3) {};
\node (9a) at (4,2) {};
\node (10) at (2,2) {};
\node (10a) at (3,1) {};
\node (11) at (1,1) {};
\end{scope}

\begin{scope}[very thick]
\draw (1)--(1a);
\draw (1x)--(1a);
\draw (1y)--(1a);
\draw (3)--(3a);
\draw (4)--(4a);
\draw (5)--(5a);
\draw (6)--(6a);
\draw (9)--(9a);
\draw (10)--(10a);

\draw (1) -- (2);
\draw (10) -- (11);

\begin{scope}[densely dotted]
\draw[green!70!black] (2) -- (3);
\draw (3) -- (4);
\draw[green!70!black](4) -- (5);
\draw (5) -- (6);
\draw[green!70!black] (6) -- (7);
\draw[green!70!black] (7) -- (8);
\draw (8) -- (9);
\draw[green!70!black] (9) -- (10);
\end{scope}

\end{scope}

\begin{scope}[us]
\draw (10) to[bend right=15] (2);
\end{scope}

\node[above] (r) at (1) {$r$};
\node[above] (t) at (2) {$t$};
\node[above] (b) at (10) {$b$};

\begin{scope}[darkred]
\node[right] (u1) at (-7,3) {$u$};
\end{scope}

\begin{scope}[fs]
\draw (3a) to[out=-70, in=-140] (6a);
\draw (5a) to[out=-70, in=-140]  (9a);
\draw (8) to[out=-100, in=-130, looseness=1.7] (10a);
\draw (1x) to[out=-30, in=-140](4a);
\end{scope}

\begin{scope}[blue, every node/.append style={font=\normalsize}]
\node[right] (l1) at (-14.5,0.3) {$\ell_1$};
\node[right] (l2) at (-11.5,-0.1) {$\ell_2$};
\node[right] (l3) at (-8,-1.3) {$\ell_3$};
\node[right] (l4) at (-5,-1.8) {$\ell_4$};
\end{scope}

\begin{scope}[green!50!black, every node/.append style={font=\normalsize}]
\node[right] (p1) at (-14,3.1) {$P_{u,\ell_1}$};
\node[right] (p2) at (-11,2.4) {$P_{u,\ell_2}$};
\node[right] (p3) at (-7.8,1.4) {$P_{u,\ell_3}$};
\node[right] (p4) at (-4.2,0.5) {$P_{u,\ell_4}$};
\end{scope}

\end{tikzpicture}
 \end{center}
\caption{The picture shows the tree $G$ (black and green), an up-link $u$, and a minimal set $F_u$ of links with $P_u \subseteq \bigcup_{\ell \in F_u} P_{\ell}$ (blue). The edges in $P_u$ are dotted and the edges in $P_{u,\ell_i}$ for $i\in\{1,2,3,4\}$ are shown in green.
}
\label{fig:order_on_F_u}
\end{figure}

\begin{lemma}\label{lem:P_u_ell_path}
Let $u\in L_{\mathrm{up}}$ and $\ell \in F_u$.
Then $P_{u,\ell}$ is nonempty and the edge set of a path.
\end{lemma}
\begin{proof}
The minimality of $F_u$ immediately implies $P_{u,\ell} \ne \emptyset$.
Let $e_1,e_2,e_3 \in P_u\cap P_{\ell}$ be three distinct edges that appear in this order on the path $(V_u,P_u)$.
If $e_1, e_3 \in P_{u,\ell}$, then either $e_2$ is also contained in $P_{u,\ell}$ or there is a link $\bar \ell \in F_u \setminus \{\ell\}$ with $e_2\in P_{\bar \ell}$.
In the latter case, $e_2$ is contained in $P_u \cap P_{\bar \ell}$ which is the edge set of a subpath of $(V_u,P_u)$.
Because $e_1$ and $e_3$ are not contained in $P_{\bar \ell}$, 
we have $P_u \cap P_{\bar \ell} \subsetneq P_u \cap P_{\ell}$, contradicting
the fact that $P_{u, \bar \ell}$ is nonempty.
Hence, whenever $e_1$ and $e_3$ are contained in $P_{u,\ell}$, then all edges that appear between $e_1$ and $e_3$ on the path $(V_u,P_u)$ are also contained in $P_{u,\ell}$. This shows that $P_{u,\ell}\subseteq P_u$ is the edge set of a path.
\end{proof}

The edge sets $P_{u,\ell}$ with $\ell\in F_u$ are pairwise disjoint by their definition.
For $\ell_1,\ell_2\in F_u$, we define $\ell_1 \prec_u \ell_2$ if and only if the edges in $P_{u,\ell_1}$ appear before the edges of $P_{u,\ell_2}$ on the $t$-$b$ path in $G$.
By Lemma~\ref{lem:P_u_ell_path}, this order is well-defined.
An alternative characterization of the same link order is that, for $\ell_1, \ell_2 \in F_u$, we have $\ell_1 \prec_u \ell_2$ if and only if $\apex(\ell_1)$ is a strict ancestor of $\apex(\ell_2)$. (This characterization follows from Lemma~\ref{lem:basic_arc_properties}~\ref{item:comparable_links_have_strict_apex_ancestorship} below.)

The dependency graph for a set $U\subseteq L_{\mathrm{up}}$ of up-links
is a directed graph with vertex set $F$.
For every up-link $u\in U$, it contains a set $A_u$ of arcs defined as follows.
For $u\in U$, let $\ell_1 \prec_u \ell_2 \prec_u \dots \prec_u \ell_q$ be the links in $F_u$. Then 
\begin{equation*}
 A_u \coloneqq \bigl\{ (\ell_i, \ell_{i+1}) \colon i\in\{1,\dots, q-1\}\bigr\}\enspace.
\end{equation*}
The arc set of the dependency graph for $U$ is the disjoint union of the sets $A_u$ for all $u\in U$.
\smallskip

This construction immediately implies that, for every up-link $u\in U$, there is one connected component $(C,A)$ of the dependency graph such that $F_u \subseteq C$, which implies that $P_u$ is covered by the links in $C$.
To prove the decomposition theorem, we will show that there exists a set $R\subseteq U$ with $w(R) \le \epsilon \cdot w(U)$ such that, for every connected component $(C,A)$ of the dependency graph of $U\setminus R$, the link set $C$ is $\lceil\sfrac{1}{\epsilon}\rceil$-thin.

First, we show some basic properties of the dependency graph that do not rely on our particular choice of the sets $F_u$, except for them being a minimal link set covering $P_u$. We provide self-contained proofs here but remark that some of these properties, in particular property~\ref{item:branching}, have been shown already in \cite{cohen_2013_approximation}. Subsequently, in Section~\ref{sec:thinCompsDepGraph}, we exploit our particular choice of the sets $F_u$ and show how they allow for obtaining property~\ref{item:thinness}.

\begin{lemma}\label{lem:basic_arc_properties}
Let $u\in L_{\mathrm{up}}$ and $(\ell_1,\ell_2)\in A_u$.
Then
\begin{enumerate}
\item\label{item:comparable_links_have_strict_apex_ancestorship} $\apex(\ell_1)$ is a strict ancestor of $\apex(\ell_2)$ in the tree $G$, and \label{item:ancestor_relation_apex}
\item $P_{u,\ell_1}$ is the edge set of a subpath of the $\apex(\ell_1)$-$\apex(\ell_2)$ path in $G$. \label{item:position_P_u_ell}
\end{enumerate}
\end{lemma}
\begin{proof}
Let $u=\{t,b\}\in U$ be the up-link with $(\ell_1,\ell_2)\in A_u$,
where $t$ is an ancestor of $b$ in the tree $G$.
Let $\{a_1,b_1\} \in P_{u,\ell_1}$ and $\{a_2,b_2\} \in P_{u,\ell_2}$.
Without loss of generality we may assume that $a_1$ is an ancestor of $b_1$ and $a_2$ is an ancestor of $b_2$. (Because $a_1\neq b_1$ and $a_2\neq b_2$, they are actually strict ancestors.)
By the definition of the order $\prec_u$ and the arc set $A_u$, the edge
$\{a_1,b_1\}$ appears before the edge $\{a_2,b_2\}$ on the $t$-$b$ path in $G$.
Because $t$ is an ancestor of $b$, we conclude that $a_1$ is an ancestor of $a_2$.
Moreover, because the edge $\{a_1,b_1\}$ is covered by the link $\ell_1$, the vertex $\apex(\ell_1)$ is an ancestor of $a_1$.
The link $\ell_2$ covers $\{a_2,b_2\}$ but it does not cover $\{a_1,b_1\}$ due to the definition of $P_{u,\ell_1}$.
Therefore, $\apex(\ell_2)$ must be a descendant of $b_1$.
We have shown that the vertex $\apex(\ell_1)$ is an ancestor of $a_1$, which in turn is a strict ancestor of $b_1$, which is an ancestor of $\apex(\ell_2)$. Hence, $\apex(\ell_1)$ is a strict ancestor of $\apex(\ell_2)$.
Moreover, this shows that every edge $\{a_1,b_1\}\in P_{u,\ell_1}$ lies on the $\apex(\ell_1)$-$\apex(\ell_2)$ path in $G$.
By Lemma~\ref{lem:P_u_ell_path}, this implies \ref{item:position_P_u_ell}.
\end{proof}

\begin{lemma}\label{lem:incoming_arcs}
Let $u\in L_{\mathrm{up}}$ be an up-link and let $(\ell_1, \ell_2)\in A_u$.
Let $e\in E$ be the last edge of the $r$-$\apex(\ell_2)$ path in $G$.
Then $u$ covers $e$.
\end{lemma}
\begin{proof}
By Lemma~\ref{lem:basic_arc_properties}, 
$\apex(\ell_1)$ is an ancestor of $\apex(\ell_2)$ and 
the edges in the nonempty set $P_{u,\ell_1}\subseteq P_u$ lie on the $\apex(\ell_1)$-$\apex(\ell_2)$ path in $G$.
Thus, $V_u$ contains at least one strict ancestor of $\apex(\ell_2)$.
Because $P_u \cap P_{\ell_2}$ is nonempty,
$V_u$ contains at least one descendant of $\apex(\ell_2)$.
Using that $(V_u,P_u)$ is a path in the tree $G$, we conclude $e\in P_u$.
\end{proof}

\begin{lemma}\label{lem:branching}
Let $U\subseteq L_{\mathrm{up}}$ be a set of up-links such that the sets $P_u$ for $u\in U$ are pairwise disjoint.
Then the dependency graph of $U$ is a branching.
\end{lemma}
\begin{proof}
By Lemma~\ref{lem:basic_arc_properties}~\ref{item:ancestor_relation_apex}, the dependency graph does not contain any directed cycle. 
Hence, it remains to show that every link has at most one incoming arc.
Let $\ell \in F$ and let $e$ be the last edge of the $r$-$\apex(\ell)$ path in $G$. 
Because the edge sets $P_u$ with $u\in U$ are pairwise disjoint, there is at most one up-link $u\in U$ that covers $e$.
By Lemma~\ref{lem:incoming_arcs}, every arc entering $\ell$ in the dependency graph is contained in $A_u$.
Because $A_u$ is the arc set of a directed path, $\ell$ has at most one incoming arc.
\end{proof}

To prove property~\ref{item:thinness}, we rely on the following lemma, which is a special case of it. We later use this special case to obtain the general statement.

\begin{lemma}\label{lem:thinness_of_F_u}
$F_u$ is $2$-thin for any $u\in L_{\mathrm{up}}$.
\end{lemma}
\begin{proof}
Let $v\in V$.
With the goal of deriving a contradiction, suppose there exist distinct links $\ell_1,\ell_2,\ell_3\in F_u$ with $v\in V_{\ell_1} \cap V_{\ell_2} \cap V_{\ell_3}$.
We may assume $\ell_1 \prec_u \ell_2 \prec_u \ell_3$ without loss of generality.
Let $z\in V_u$ be the vertex of $V_u$ that is closest to $v$ in the tree $G$. (In particular, if $v\in V_u$, then $z=v$.)
Note that any path containing both $v$ and a vertex of $V_u$ must go through $z$. Hence, $z\in V_{\ell_i}$ for $i\in \{1,2,3\}$.
Because $P_{\ell_i} \cap P_u$ is the edge set of a subpath of $(V_u,P_u)$ with vertex set $V_{\ell_i} \cap V_u$ for all $i\in\{1,2,3\}$, and all of these paths contain the vertex $z$, also $(P_{\ell_1} \cap P_u) \cup (P_{\ell_3} \cap P_u)$ is the edge set of a subpath of $(V_u,P_u)$. This subpath covers the edges in $P_{u,\ell_1}$ as well as the edges in $P_{u,\ell_3}$.
Because $\ell_1 \prec_u \ell_2 \prec_u \ell_3$, this implies that it also covers the edges in $P_{u,\ell_2}$, contradicting the definition of the nonempty set $P_{u,\ell_2}$.
\end{proof}

\subsection{Thin components and the dependency graph}\label{sec:thinCompsDepGraph}

Most properties of the dependency graph shown so far were primarily properties of a single set $A_u$ (or immediate consequences of these). Therefore, only minimality of the sets $F_u$ was necessary to show them. We now move toward more global results on the connected components of the dependency graph by exploiting our particular choice of the sets $F_u$, with the goal to show property~\ref{item:thinness}.
To this end, let $U\subseteq L_{\mathrm{up}}$ be a set of up-links such that the sets $P_u$ with $u\in U$ are pairwise disjoint.
We fix a connected component $(C,A)$ of the dependency graph of $U$.
By Lemma~\ref{lem:branching}, the connected component $(C,A)$ is an arborescence.
Even though we do not exploit this later, we note that one can show that any distinct links $\ell_1, \ell_2\in C$ have distinct apexes, i.e., $\apex(\ell_1)\neq \apex(\ell_2)$.\footnote{This will for example follow from Lemma~\ref{lem:non_ancestors_do_not_touch_the_same_vertex}. Indeed, with the goal of deriving a contradiction, assume that there are two distinct links $\ell_1, \ell_2 \in C$ with same apex. By Lemma~\ref{lem:non_ancestors_do_not_touch_the_same_vertex}, $\ell_1$ and $\ell_2$ must have an ancestry relationship in $(C,A)$ because their common apex is in $V_{\ell_1}\cap V_{\ell_2}$. This ancestry relationship is strict because $\ell_1\neq \ell_2$. Moreover, by Lemma~\ref{lem:basic_arc_properties}~\ref{item:comparable_links_have_strict_apex_ancestorship}, any parent-child relationship (and therefore any strict ancestry relationship) between two links in $(C,A)$ implies that the apex of the parent is a strict ancestor of the apex of the child. This contradicts $\apex(\ell_1)=\apex(\ell_2)$.}

The next lemma is a crucial step toward property~\ref{item:thinness}, as it shows that links in the same component whose paths $(V_{\ell},P_{\ell})$ have a common vertex, must have an ancestry relationship in the dependency graph. 
Note that this lemma together with Lemma~\ref{lem:thinness_of_F_u} already imply a weaker version of property~\ref{item:thinness}, namely that if the arc set of every directed path in $(C,A)$ has nonempty intersection with $A_u$ for at most $k$ up-links $u\in U$, then $C$ is $2k$-thin.%
\footnote{This weaker version of property~\ref{item:thinness} is already sufficient to prove our main result.
However, to guarantee correctness through this weaker property, it would not suffice to consider $\lceil \sfrac{2}{\epsilon}\rceil$-thin components in Algorithm~\ref{algo:relative_greedy}, but more general components are needed instead, for example $\lceil \sfrac{4}{\epsilon}\rceil$-thin ones.}
Indeed, Lemma~\ref{lem:non_ancestors_do_not_touch_the_same_vertex} below implies that, for any vertex $v\in V$, the links $\ell\in C$ with $v\in V_\ell$ must lie on a directed path in $(C,A)$.
Finally, Lemma~\ref{lem:thinness_of_F_u} shows that for each of the at most $k$ links $u\in U$ for which $A_u$ intersects that path, there has at most $2$ links $\ell\in F_u$ satisfying $v\in P_\ell$.
After proving Lemma~\ref{lem:non_ancestors_do_not_touch_the_same_vertex}, we strengthen this reasoning to obtain property~\ref{item:thinness}, which is tight.
\begin{lemma}\label{lem:non_ancestors_do_not_touch_the_same_vertex}
Let $\ell_1,\ell_2 \in C$ with $V_{\ell_1} \cap V_{\ell_2} \ne \emptyset$.
Then $\ell_1$ and $\ell_2$ have an ancestry relationship in the arborescence $(C,A)$, i.e., either $\ell_1$ is an ancestor of $\ell_2$ or $\ell_2$ is an ancestor of $\ell_1$.
\end{lemma}
\begin{proof}
Let $v\in V_{\ell_1} \cap V_{\ell_2}$. Then $\apex(\ell_1)$ and $\apex(\ell_2)$ are ancestors of $v$ in the tree $G$.
Therefore, $\apex(\ell_1)$ and $\apex(\ell_2)$ have an ancestry relation in $G$, say $\apex(\ell_1)$ is an ancestor of $\apex(\ell_2)$.
Thus, $\apex(\ell_2)$ lies on the $\apex(\ell_1)$-$v$ path in $G$.
Because $v\in V_{\ell_1}$, this implies $\apex(\ell_2) \in V_{\ell_1}$.
See the left part of Figure~\ref{fig:proof_ancestry_relation}.

For the sake of deriving a contradiction, suppose that $\ell_1$ and $\ell_2$ have no ancestry relation in the arborescence $(C,A)$.
Then the path from the root of $(C,A)$ to $\ell_2$ does not contain $\ell_1$. By Lemma~\ref{lem:basic_arc_properties}~\ref{item:ancestor_relation_apex}, there is an arc $a=(\ell, \bar \ell)$ on this path such that $\apex(\ell)$ is a strict ancestor of $\apex(\ell_1)$ and $\apex(\bar \ell)$ is a descendant of $\apex(\ell_1)$. 
Then $\apex(\bar \ell)$ lies on the $\apex(\ell_1)$-$\apex(\ell_2)$ path in $G$.
This implies $\apex(\bar \ell) \in V_{\ell_1}$ because $\apex(\ell_2)\in V_{\ell_1}$.

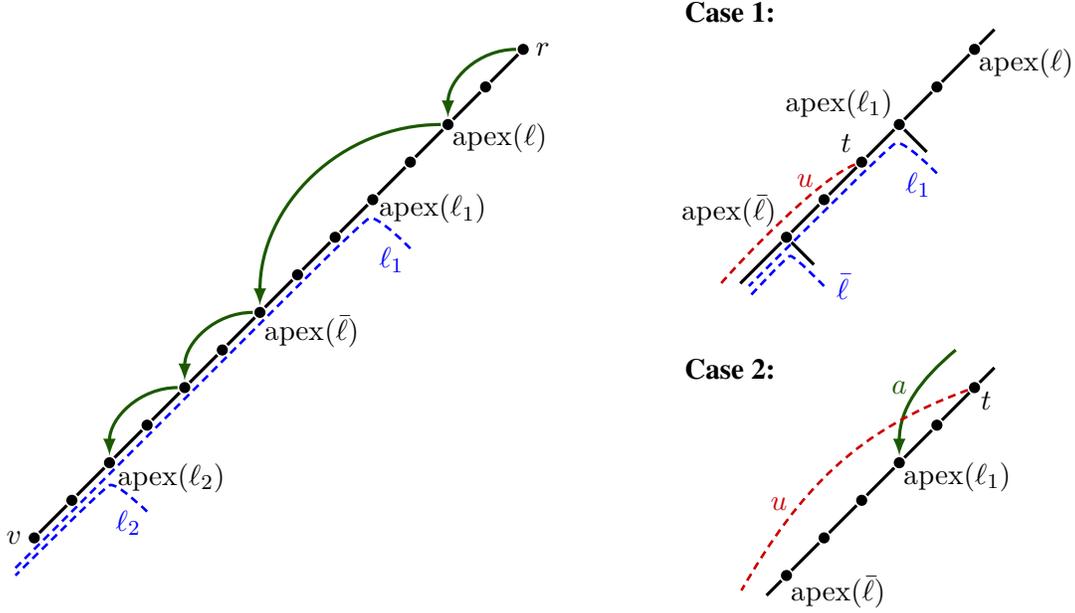
\begin{figure}[t]
\begin{center}
\begin{tikzpicture}[scale=0.5]

\tikzset{
fs/.style={line width=1pt, blue, densely dashed},
us/.style={line width=1pt, darkred, densely dashed}
}

\begin{scope}[every node/.style={thick,draw=black,fill=black,circle,minimum size=0pt, inner sep=1.2pt, outer sep=1pt}]

\node (0) at (15,15) {};
\node (1) at (14,14) {};
\node (2) at (13,13) {};
\node (3) at (12,12) {};
\node (4) at (11,11) {};
\node (5) at (10,10) {};
\node (6) at (9,9) {};
\node (7) at (8,8) {};
\node (8) at (7,7) {};
\node (9) at (6,6) {};
\node (10) at (5,5) {};
\node (11) at (4,4) {};
\node (12) at (3,3) {};
\node (13) at (2,2) {};

\end{scope}

\begin{scope}[darkgreen,very thick, ->, >=latex, in=90, out=180]
\draw (0) to (2);
\draw (2) to (7);
\draw (7) to (9);
\draw (9) to (11);
\end{scope}

\node[right=1pt] (r) at (0) {$r$};
\node[below right=-3pt,yshift=2pt] (r) at (2) {$\apex(\ell)$};
\node[below right=-2.0pt,yshift=4.5pt] (r) at (4) {$\apex(\ell_1)$};
\node[below right=-3pt,yshift=1pt] (r) at (7) {$\apex(\bar \ell)$};
\node[below right=-1pt,yshift=3pt] (r) at (11) {$\apex(\ell_2)$};
\node[left=1pt] (r) at (13) {$v$};

\begin{scope}[very thick]
 \draw (0) -- (1) -- (2) -- (3) -- (4) -- (5) --(6) -- (7) -- (8) --(9) --(10) --(11) -- (12) -- (13);
\end{scope}

\begin{scope}[fs]
\draw plot[smooth, tension=0.2] coordinates {(12,9.7)(10.9,10.5)(8,7.7)(1.5,1.2)};
\draw plot[smooth, tension=0.2] coordinates {(5,2.7)(4.0,3.4)(1.5,1)};
\end{scope}

\node[blue] () at (11.5,9.4) {$\ell_1$};
\node[blue] () at (4.5,2.4) {$\ell_2$};

\begin{scope}[shift={(14,2)}]

\node () at (6.5,14) {\textbf{Case 1:}};

\begin{scope}[every node/.style={thick,draw=black,fill=black,circle,minimum size=0pt, inner sep=1.2pt, outer sep=1pt}]

\node (2) at (13,13) {};
\node (3) at (12,12) {};
\node (4) at (11,11) {};
\node (5) at (10,10) {};
\node (6) at (9,9) {};
\node (7) at (8,8) {};
\end{scope}
\node (1) at (13.8,13.8) {};
\node (8) at (6.5,6.5) {};
\node (4a) at (12,10) {};
\node (7a) at (9,7) {};

\begin{scope}[fs]
\draw plot[smooth, tension=0.2] coordinates {(12,9.7)(10.9,10.5)(8,7.7)(7,6.7)};
\draw plot[smooth, tension=0.2] coordinates {(9,6.7)(8.1,7.5)(7,6.4)};
\end{scope}

\begin{scope}[us]
\draw[out=200,in=50, looseness=0.5] (5) to (6.2,6.7);
\end{scope}

\node[below right=-3pt,yshift=2pt] (r) at (2) {$\apex(\ell)$};
\node[above left=-2pt] (r) at (4) {$\apex(\ell_1)$};
\node[above left] (r) at (7) {$\apex(\bar \ell)$};
\node[above left] (r) at (5) {$t$};

\node[blue] () at (11.5,9.4) {$\ell_1$};
\node[blue] () at (9.5,6.7) {$\bar \ell$};
\node[darkred] () at (8.5,9.5) {$u$};

\begin{scope}[very thick]
 \draw (1) --(2) -- (3) -- (4) -- (5) --(6) -- (7) -- (8);
 \draw (4) -- (4a);
 \draw (7) -- (7a);
\end{scope}
\end{scope}

\begin{scope}[shift={(14,-7)}]

\node () at (6.5,13.5) {\textbf{Case 2:}};

\begin{scope}[every node/.style={thick,draw=black,fill=black,circle,minimum size=0pt, inner sep=1.2pt, outer sep=1pt}]
\node (2) at (13,13) {};
\node (3) at (12,12) {};
\node (4) at (11,11) {};
\node (5) at (10,10) {};
\node (6) at (9,9) {};
\node (7) at (8,8) {};
\end{scope}
\node (1) at (13.8,13.8) {};
\node (8) at (7.2,7.2) {};

\begin{scope}[darkgreen,very thick, ->, >=latex, in=90, out=220]
\draw (12.5,14) to (4);
\end{scope}

\begin{scope}[us]
\draw[out=200,in=60, looseness=1] (2) to (6.8,7.6);
\end{scope}

\node[below right=-3pt, yshift=2pt] (r) at (4) {$\apex(\ell_1)$};
\node[below right=-3pt, yshift=2pt] (r) at (7) {$\apex(\bar \ell)$};
\node[below right=-2pt, yshift=1pt] (r) at (2) {$t$};

\begin{scope}[very thick]
 \draw (1) -- (2) -- (3) -- (4) -- (5) --(6) -- (7) -- (8);
\end{scope}

\node[darkred] () at (7.8,9.9) {$u$};
\node[darkgreen] () at (11,13) {$a$};

\end{scope}

\end{tikzpicture}
 \end{center}
\caption{Illustration of the proof of Lemma~\ref{lem:non_ancestors_do_not_touch_the_same_vertex}.
Here, a dark green arc $(\apex(\ell),\apex(\bar \ell))$ in the picture represents the arc $(\ell,\bar \ell)$ in the dependency graph.
Black vertices and edges show parts of the tree $G$ and links are drawn dashed.
}\label{fig:proof_ancestry_relation}
\end{figure}

Let $u\in U$ with $(\ell,\bar \ell)\in A_u$.
Let $t$ denote the endpoint of $u$ that is closer to the root of the tree $G$. We distinguish two cases.
In the first case, we assume that $t$ is a descendant of $\apex(\ell_1)$. (See top illustration on right-hand side of Figure~\ref{fig:proof_ancestry_relation}.)
Then the edges of the $t$-$\apex(\bar \ell)$ path are covered by $\ell_1$ because  $\apex(\bar \ell)\in V_{\ell_1}$.
Thus, by Lemma~\ref{lem:basic_arc_properties}~\ref{item:position_P_u_ell}, $\ell_1$ covers all edges in $P_{u,\ell}$ and hence
\begin{equation}\label{eq:better_covering}
P_u \subseteq \bigcup_{\ell' \in (F_u \setminus \{\ell\}) \cup \{\ell_1\}} P_{\ell'} \enspace.
\end{equation}
Because $t$ is a descendant of $\apex(\ell_1)$, it is also a descendant of $\apex(\ell)$ and therefore $\ell$ is the first link in $F_u$ with respect to the order $\prec_u$ by Lemma~\ref{lem:basic_arc_properties}.
We conclude that for every $\ell'\in (F_u \setminus \{\ell\}) \cup \{\ell_1\}$, the vertex $\apex(\ell')$ is a strict descendant of $\apex(\ell)$. Because of \eqref{eq:better_covering}, this implies that the vertex $v_u$ in the construction of $F_u$ must be a strict descendant of $\apex(\ell)$, contradicting $\ell \in F_u$.

Now consider the remaining second case, where $t$ is a strict ancestor of $\apex(\ell_1)$. (See bottom illustration on right-hand side of Figure~\ref{fig:proof_ancestry_relation}.)
Because $\ell_1$ and $\ell_2$ have no ancestry relation, $\ell_1$ is not the root of $(C,A)$ and hence $\ell_1$ has an incoming arc $a\in A$.
By Lemma~\ref{lem:incoming_arcs}, $P_u$ covers the last edge of the $r$-$\apex(\bar \ell)$ path in $G$.
In particular, both the descendant $\apex(\bar \ell)$ of $\apex(\ell_1)$ and the strict ancestor $t$ of $\apex(\ell_1)$ are contained in $V_u$.
This implies that $u$ is the unique up-link in $U$ that covers the last edge of the $r$-$\apex(\ell_1)$ path in $G$.
By Lemma~\ref{lem:incoming_arcs}, we conclude that the incoming arc $a$ of $\ell_1$ is contained in $A_u$.
Thus, we have $\ell_1, \ell,\bar \ell \in F_u$.
Because $\apex(\ell)$ is an ancestor of $\apex(\ell_1)$, which in turn is an ancestor of $\apex(\bar \ell)$, we have $\ell \prec_u \ell_1 \prec_u \bar \ell$ by Lemma~\ref{lem:basic_arc_properties}~\ref{item:ancestor_relation_apex}.
This contradicts $(\ell, \bar \ell)\in A_u$.
\end{proof}

Recall that in order to prove property~\ref{item:thinness} of the dependency graph, we need to give an upper bound on $|\{\ell\in C : v\in V_{\ell}\}|$ for all $v\in V$.
The next lemma establishes this upper bound for every vertex $v\in V$ that is the apex of some link $\ell\in C$.

\begin{lemma}\label{lem:path_length_induction}
Let $\ell \in C$ and let $H_{\ell}$ be the arc set of the path from the root of the arborescence $(C,A)$ to $\ell$.
Then
\begin{equation*}
\left|\left\{ \bar \ell \in C\setminus \{\ell\} \colon \apex(\ell) \in V_{\bar \ell} \right\}\right|
\ \le\ \left|\left\{ u\in U \colon H_{\ell} \cap A_u \ne \emptyset \right\}\right|\enspace.
\end{equation*}
\end{lemma}
\begin{proof}
Let $k\coloneqq |\{ u\in U \colon H_{\ell} \cap A_u \ne \emptyset \}|$. We prove the lemma by induction on $k$.
If $k=0$, then the link $\ell$ is the root of the arborescence $(C,A)$.
This implies by Lemma~\ref{lem:basic_arc_properties}~\ref{item:ancestor_relation_apex} that $\apex(\bar \ell)$ is a strict descendant of $\apex(\ell)$ for every $\bar \ell \in C\setminus \{\ell\}$. Hence, in this case $\{ \bar \ell \in C\setminus \{\ell\} \colon \apex(\ell) \in V_{\bar \ell} \} = \emptyset$, as desired.

Now suppose $k>0$. 
By Lemma~\ref{lem:non_ancestors_do_not_touch_the_same_vertex} (and Lemma~\ref{lem:basic_arc_properties}~\ref{item:ancestor_relation_apex}),
every link $\bar \ell \in C$ with $\apex(\ell)\in V_{\bar{\ell}}$
is an ancestor of $\ell$ in $(C,A)$.
Because $k>0$, the link $\ell$ has an incoming arc $a\in A$. 
Let $u\in U$ be the unique up-link with $a\in A_u$.
Let $\ell_{1}\in F_u$ be the first link on the directed path $(F_u,A_u)$ in the arborescence $(C,A)$.
Then for every ancestor $\bar \ell$ of $\ell$ in the arborescence $(C,A)$, and hence for every link $\bar \ell \in C$ with $\apex(\ell)\in V_{\bar{\ell}}$, we have either
\begin{enumerate}
\item\label{item:links_induction} $\bar \ell$ is a strict ancestor of $\ell_1$ in $(C,A)$, or
\item\label{item:link_F_u} $\bar \ell \in F_u$.
\end{enumerate}
Consider a strict ancestor $\bar \ell$ of $\ell_1$.
By Lemma~\ref{lem:basic_arc_properties}~\ref{item:ancestor_relation_apex}, $\apex(\bar \ell)$ is an ancestor of $\apex(\ell_1)$, which in turn is an ancestor of $\apex(\ell)$ in the tree $G$.
Therefore, if $\apex(\ell) \in V_{\bar{\ell}}$, then also $\apex(\ell_1) \in V_{\bar{\ell}}$. 
By the inductive hypothesis applied to $\ell_1$, this implies that there are at most $k-1$ links in $\bar \ell \in C\setminus \{\ell_1\}$ that fulfill both \ref{item:links_induction} and $\apex(\ell) \in V_{\bar{\ell}}$.
Thus, it suffices to show that there is at most one link $\bar \ell \in F_u\setminus \{\ell\}$ with $\apex(\ell)\in V_{\bar \ell}$.
This holds because by Lemma~\ref{lem:thinness_of_F_u}, there are at most two links $\bar \ell \in F_u$ with $\apex(\ell) \in V_{\bar \ell}$, one of which is $\ell$.
\end{proof}

Finally, we are ready to prove property~\ref{item:thinness} of the dependency graph.

\begin{lemma}\label{lem:thinness}
Let $k\in \mathbb{Z}_{\geq 0}$. If for every path in $(C,A)$ with arc set $H\subseteq A$, we have
\begin{equation*}
\left|\left\{ u\in U \colon H \cap A_u \ne \emptyset \right\}\right| \le k\enspace,
\end{equation*}
then $C$ is $(k+1)$-thin.
\end{lemma}
\begin{proof}
We need to show that, for every $v\in V$, there are at most $k+1$ links $\ell \in C$ with $v\in V_{\ell}$.
Let $z\in V$ be the last vertex on the $r$-$v$ path in $G$ that is the apex of some link in $C$, and let $\ell_{z}\in C$ be a link with $\apex(\ell_{z})=z$.
For every link $\ell\in C$ with $v \in V_{\ell}$, the vertex $\apex(\ell)$ is an ancestor of $v$. 
Thus, by the choice of $z$, we have $z \in V_\ell$ for each such link $\ell$.
Hence, it suffices to show that there are at most $k+1$ links $\ell\in C$ with $z\in V_{\ell}$.
By Lemma~\ref{lem:path_length_induction}, there are at most $k$ links $\ell\in C\setminus \{\ell_{z}\}$ with $z \in V_{\ell}$ and hence at most $k+1$ links $\ell \in C$ with $z \in V_{\ell}$.
\end{proof}

\subsection{Proof of the decomposition theorem}

We are now ready to complete the proof of the decomposition theorem, which we restate here for convenience. See Figure~\ref{fig:example_choice_R} for an illustration of the proof.

\decompositionTheorem*

\begin{proof}
Let $k\coloneqq \lceil\sfrac{1}{\epsilon}\rceil$.
We start by defining, independently for every connected component $(C,A)$ of the dependency graph of $U$, a labeling $c\colon A\to \mathbb{Z}_{\geq 0}$, where arcs contained in the same set $A_u$ with $u\in U$ will have the same label. (Hence, this can be interpreted as a labeling of the sets $A_u$ as we did in our brief description in Figure~\ref{fig:example_choice_R}.)
For all arcs $a$ in a path $(F_u,A_u)$ (with $u\in U$) that start at the root of the arborescence $(C,A)$, we set $c(a)\coloneqq 0$.
For a path $(F_u,A_u)$ that starts at a link $\ell$ that is not the root of $(C,A)$, let $j\in \mathbb{Z}_{\geq 0}$ be the label of the incoming arc of $\ell$. Then we set $c(a) = j+1$ for all $a \in A_u$.
Because $(C,A)$ is an arborescence and we can consider the paths $(A_u,F_u)$ in an order of increasing distance of their start point from the root of $(C,A)$, this indeed defines a labeling $c\colon A \to \mathbb{Z}_{\geq 0}$.

For $i\in \{0,\dots,k-1\}$, let $R_i \subseteq U$ be the set of up-links in $U$ for which the arcs in $A_u$ have a label $j$ with $j\equiv i \pmod{k}$.
Then $\{R_0,R_1,\dots,R_{k-1}\}$ is a partition of $U$.
Hence, there exists some $i\in \{0,\dots,k-1\}$ such that $w(R_i) \le \sfrac{w(U)}{k}  \le \epsilon \cdot w(U)$, and we set $R=R_i$.

This completes the construction of $R\subseteq U$ with $w(R) \le \epsilon \cdot w(U)$.
We choose the partition $\Cscr$ of $F$ to be the collection of the vertex sets of the connected components of the dependency graph of $U\setminus R$.
The dependency graph of $U\setminus R$ arises from the dependency graph of $U$ by deleting the arcs of each $k$-th label, starting with label $i$.
Therefore, by the construction of the labeling $c\colon A\to \mathbb{Z}_{\geq 0}$, every path in the dependency graph of $U\setminus R$ has nonempty intersection with $A_u$ for at most $k-1$ links in $U\setminus R$.
By Lemma~\ref{lem:thinness}, this implies that all elements of $\Cscr$ are $k$-thin. 
Finally, we observe that by the definition of $\Cscr$, we have that, for every $u\in U\setminus R$, there exists some $C\in \Cscr$ with $F_u\subseteq C$. This shows property~\ref{item:links_outside_R_covered}.
\end{proof}
\section{Finding optimal thin components}\label{sec:dynamic_program}

In this section we prove Lemma~\ref{lem:finding_optimal_thin_components}, i.e., we prove that, for any constant $k\in \mathbb{Z}_{\geq 0}$, we can efficiently find a $k$-thin set $C\subseteq L$ that minimizes $\sfrac{w(C)}{w(\Drop_U(C))}$. To this end, we compute the optimal ratio 
\begin{equation}
\rho^* \coloneqq \min \left\{\frac{w(C)}{w(\Drop_U(C))} : C\subseteq L \text{ is $k$-thin}\right\}\label{eq:rhoStarAsMinRatioProb}
\end{equation}
and a corresponding minimizer through a binary search procedure that relies on a dynamic program to decide whether some value $\rho\in \mathbb{R}$ is larger or smaller than $\rho^*$. (We recall that $\sfrac{w(C)}{w(\Drop_U(C))}$ is interpreted as $\infty$ whenever $w(\Drop_U(C))=0$, which, due to strictly positive link weights, is equivalent to $C=\emptyset$.)

\subsection{Reducing to slack maximization}

The question of whether some value $\rho \in \mathbb{R}$ is larger or smaller than $\rho^*$ naturally reduces to maximizing the following slack function $\slack_{\rho}(C)$ over all $k$-thin sets $C\subseteq L$:
\begin{equation*}
\slack_{\rho}(C) \coloneqq \rho \cdot w\bigl(\Drop_U(C)\bigr) - w(C)\enspace.
\end{equation*}
More precisely, we have the following simple yet very helpful equivalence, which immediately follows from the definition of $\slack_{\rho}$.
\begin{observation}\label{obs:basicSlackRatioRel}
Let $\rho\in \mathbb{R}$ and $C\subseteq L$ with $C\neq \emptyset$. Then $\slack_{\rho}(C)\geq 0$ if and only if $\frac{w(C)}{w(\Drop_U(C))} \leq \rho$.
\end{observation}

Hence, the question how a given $\rho\in \mathbb{R}$ compares to $\rho^*$, which is what we need to apply binary search, reduces to maximizing $\slack_{\rho}(C)$ over nonempty $k$-thin sets $C\subseteq L$, as formalized below.
\begin{observation}\label{obs:binSearchToDP}
Let $\rho\in \mathbb{R}$. Then the following two statements are equivalent:
\begin{enumerate}
\item\label{item:rhoAtLeastRhoStar} $\rho \geq \rho^*$.
\item\label{item:DPProbAtLeastZero} $\max\{\slack_{\rho}(C) \colon C\subseteq L \text{ is $k$-thin and } C\neq \emptyset\} \geq 0$.
\end{enumerate}
\end{observation}
\begin{proof}
We have $\rho \geq \rho^*$ if and only if there is a $k$-thin set $C\subseteq L$ with $\sfrac{w(C)}{w(\Drop_U(C))} \leq \rho$. Notice that $C$ must be nonempty, for otherwise we would have $\Drop_U(C)=\emptyset$ and $\sfrac{w(C)}{w(\Drop_U(C))}$ would have been interpreted as $\infty$, which violates finiteness of $\rho$. By Observation~\ref{obs:basicSlackRatioRel}, we thus obtain that $\rho \geq \rho^*$ if and only if there is a nonempty $k$-thin set $C\subseteq L$ with $\slack_{\rho}(C)\geq 0$, as desired.
\end{proof}

Hence, to compute the value of $\rho^*$ by binary search,
it suffices to have an algorithm for the maximization problem in point~\ref{item:DPProbAtLeastZero} of Observation~\ref{obs:binSearchToDP}.
In Section~\ref{sec:slack_max}, we describe a dynamic program that solves this maximization problem (or decides that the maximum $\slack_{\rho}(C)$ is negative).
The result of our dynamic program is summarized in the lemma below.
\begin{lemma}\label{lem:slack_maximization}
Let $k\in \mathbb{Z}_{\geq 0}$ be a constant. 
Given a number $\rho\in \mathbb{R}$, a WTAP instance $(G=(V,E),L,w)$, and $U\subseteq L_{\mathrm{up}}$  such that the edge sets $P_u$ for $u\in U$ are pairwise disjoint,
we can in polynomial time compute a $k$-thin set $\overline{C}\subseteq L$ that maximizes $\slack_{\rho}(C)$
over all $k$-thin sets $C\subseteq L$. Moreover, if there is a nonempty maximizer, then $\overline{C}\neq\emptyset$.
\end{lemma}
Note that Lemma~\ref{lem:slack_maximization} indeed implies that, for any $\rho\in \mathbb{R}$, we can decide in polynomial time whether $\max\{\slack_{\rho}(C)\colon C\subseteq L \text{ is $k$-thin and } C\neq \emptyset\}\geq 0$, due to the following. Let $\overline{C}$ be a $k$-thin set as described in Lemma~\ref{lem:slack_maximization}. If $\overline{C} \neq \emptyset$, then $\slack_{\rho}(\overline{C}) = \max\{\slack_{\rho}(C)\colon C\subseteq L \text{ is $k$-thin and }C\neq \emptyset\}$ because $\overline{C}$ maximizes $\slack_{\rho}(\overline{C})$ over all $k$-thin sets $C\subseteq L$. Otherwise, if $\overline{C}=\emptyset$, then the maximum value of $\slack_{\rho}(C)$ over all $k$-thin sets is $\slack_{\rho}(\overline{C})=0$, and because Lemma~\ref{lem:slack_maximization} would have returned a nonempty maximizer if there had been one, we have $\max\{\slack_{\rho}(C)\colon C\subseteq L \text{ is $k$-thin and }C\neq \emptyset\}<0$.

Before expanding on our dynamic program, we observe that using Lemma~\ref{lem:slack_maximization} to perform binary search over $\rho$ readily implies Lemma~\ref{lem:finding_optimal_thin_components}, which we restate below for convenience.

\LemmaOptimalComponents*
\begin{proof}
We may assume without loss of generality that $w\colon L \to \mathbb{Z}_{>0}$ is integral by scaling up the weights if necessary. Moreover, we assume $U \neq \emptyset$, as the problem is trivial otherwise.
First, observe that $0\le \rho^* \le 1$ because our algorithm can always choose $C = \{u\}$ for any $u\in U$.
As discussed, for any $\rho\in [0,1]$, we can use Lemma~\ref{lem:slack_maximization} together with Observation~\ref{obs:binSearchToDP} to decide in polynomial time whether $\rho \geq \rho^*$ or $\rho < \rho^*$. 
Moreover, if $\rho \geq \rho^*$, we obtain a nonempty $k$-thin set $C \subseteq L$ with $\slack_{\rho}(C) \geq 0$. By Observation~\ref{obs:basicSlackRatioRel}, we then have $\sfrac{w(C)}{w(\Drop_U(C))}\leq \rho$.
Thus, using binary search, we can in polynomial time determine an interval $[a,b]$ with $\rho^* \in [a,b]$ and $b-a < \sfrac{1}{w(U)^2}$, together with a $k$-thin set $C\subseteq L$ satisfying $\sfrac{w(C)}{w(\Drop_U(C))} \leq b$.

We claim that this component $C$ is a minimizer of $\sfrac{w(C)}{w(\Drop_U(C))}$ among all $k$-thin sets, as desired, i.e., $\sfrac{w(C)}{\Drop_U(C)}=\rho^*$.
We suppose $\sfrac{w(C)}{\Drop_U(C)}>\rho^*$ and derive a contradiction.
Let $C^*\subseteq L$ be such that $\sfrac{w(C^*)}{\Drop_U(C^*)}=\rho^*$.
Because $w$ is integral, we obtain the following contradiction:
\[ \frac{1}{w(U)^2}\ >\ b-a\ \ge\
\frac{w(C)}{w(\Drop_U(C))} - \frac{w(C^*)}{w(\Drop_U(C^*))}
\ \ge\ \frac{1}{w(\Drop_U(C)) \cdot w(\Drop_U(C^*))} \ \ge\ \frac{1}{w(U)^2} \enspace,
\]
where the penultimate fraction in the above chain of inequalities has a non-zero denominator because $\sfrac{w(C)}{w(\Drop_U(C))} - \sfrac{w(C^*)}{w(\Drop_U(C^*))} = \sfrac{w(C)}{w(\Drop_U(C))} - \rho^* > 0$ by assumption.
Hence, the component $C$ that we computed fulfills $\sfrac{w(C)}{\Drop_U(C)}=\rho^*$.
\end{proof}

We remark that instead of binary search, as used above, one could also employ \citeauthor{megiddo_1979_combinatorial}'s~\cite{megiddo_1979_combinatorial} parametric search technique to find the value of $\rho^*$ in strongly polynomial time. This works out because our dynamic program to obtain $\overline{C}$ as described in Lemma~\ref{lem:slack_maximization} is a strongly polynomial time algorithm, where the value of $\rho$ appears linearly in each comparison that we perform during the algorithm.

\subsection{Proving Lemma~\ref{lem:slack_maximization} by dynamic programming}\label{sec:slack_max}

We now discuss a dynamic programming algorithm that computes in polynomial time a maximizer $\overline{C}\subseteq L$ of $\slack_{\rho}(C)$ over all $k$-thin sets $C\subseteq L$, with the additional property that $\overline{C}\neq \emptyset$ if there is a nonempty maximizer, thus implying Lemma~\ref{lem:slack_maximization}.

The dynamic program follows the canonical approach of going from leaves toward the root $r$ to build such a $k$-thin maximizer $\overline{C}\subseteq L$. More precisely, It computes $k$-thin link sets in subtrees of $G$ that successively get combined to eventually obtain $\overline{C}$. To formalize this approach, we use the following notation to deal with subtrees. To refer to the vertices of a subtree with root $v\in V$, we denote by $D_v \subseteq V$ the set of all descendants of $v$ in $G$.
Moreover, to refer to links (or edges) contained in a subtree, we write, for any set $X$ of links (or edges), $X[D_v]\subseteq X$ to denote all links (or edges) in $X$ with both endpoints in $D_v$.
Additionally, $\delta_X(D_v) \subseteq X$ denotes the set of links (or edges) in $X$ with exactly one endpoint in $D_v$.
Recall that the up-links $U$ have disjoint edge sets $P_u$ for $u\in U$. Hence, for any $v\in V$, the set $\delta_U(D_v)$ contains at most one up-link, and, if it does, then this up-link covers the last edge of the $r$-$v$ path in $G$.

To build up some intuition for the dynamic program, consider a vertex $v\in V$ and the subtree below this vertex, i.e., the one with vertices $D_v$. 
To better understand how partial solutions for this subtree can get extended to larger subtrees, let us first consider a $k$-thin set of links $Q\subseteq L$ such that each link $\ell\in Q$ interacts with the subtree below $v$, i.e., $\ell$ has at least one endpoint in $D_v$. 
To later extend $Q$ to a bigger subtree (i.e., to use $Q$ in the \emph{propagation step} of a dynamic program), there are only few things we need to know about $Q$. 
More precisely, let us partition $Q$ into the links $C=Q[D_v]$ with both endpoints in $D_v$ and the links $Y=\delta_U(D_v)$ with a single endpoint in $D_v$. 
See Figure~\ref{fig:table_entries_dp}.
The crucial characteristics of $Q$ that we need to know for propagation are:
\begin{itemize}
\item the set $Y$, which satisfies $|Y|\leq k$ because $Q$ is $k$-thin;
\item the slack when considering only those covered up-links that are contained in $U[D_v]$ and only accounting for the cost of links in $C$; we denote this slack by
\begin{equation*}
\slack_{\rho}(C, Y, v) \coloneqq \rho \cdot w\Bigl(\Drop_{U[D_v]}(C\cup Y)\Bigr) - w(C)\enspace;
\end{equation*}
\item if there is a link $u\in \delta_U(D_v)$, then we need to know whether the edges $P_u[D_v]$ are covered by $Q=C\cup Y$.
This information is needed to decide whether $u$ can later be dropped if $Q$ becomes part of a larger link set that covers the edges in $P_u\setminus P_u[D_v]$.
\end{itemize}

\begin{figure}[!ht]
\begin{center}
\begin{tikzpicture}[yscale=0.75, xscale=0.6]

\tikzset{
lks/.style={line width=2pt, densely dashed}
}

\tikzset{
ulks/.style={line width=1pt, densely dashed}
}

\draw[ultra thick, fill=black!20!white, opacity=0.3] plot [smooth cycle] coordinates { (3,3.5) (7.5,-1.5) (-1.2,-1)};

\node[black!60!white] () at (8.2,-1) {$D_v$};

\begin{scope}[every node/.style={thick,draw=black,fill=black,circle,minimum size=0pt, inner sep=1.3pt, outer sep=1.5pt}]

\node (1) at (6,6) {};
\node (1a) at (7,5) {};
\node (1b) at (8,4) {};
\node (1c) at (7,3) {};
\node (1d) at (8,3) {};
\node (1e) at (9,3) {};

\node (2) at (5,5) {};
\node (3) at (4,4) {};

\node (4) at (3,3) {};
\node (4a) at (4,2) {};
\node (4b) at (5,1) {};
\node (4c) at (4.4,0) {};
\node (4d) at (6,0) {};
\node (4e) at (5.4,-1) {};
\node (4f) at (7,-1) {};

\node (5) at (2,2) {};
\node (5a) at (3,1) {};
\node (6) at (1,1) {};
\node (7) at (-0.3,-0.5) {};
\node (8) at (1,-0.7) {};
\node (9) at (2.3,-0.5) {};
\end{scope}

\node[above=1pt] (r) at (1) {$r$};

\begin{scope}[very thick]
\draw (1)--(1a);
\draw (1a) -- (1b);
\draw (1b) -- (1c);
\draw (1b) -- (1d);
\draw (1b) -- (1e);

\draw (4) -- (4a);
\draw (4a) -- (4b);
\draw (4b) -- (4c);
\draw (4b) -- (4d);
\draw (4d) --(4e);
\draw (4d) -- (4f);

\draw (5)--(5a);

\draw (1) --(2);
\draw (2) -- (3);
\draw (3) --(4);
\draw (4) --(5) --(6);

\draw (6) -- (7);
\draw (6) -- (8);
\draw (6) -- (9);

\end{scope}

\node[left] (v) at (4) {$v$};
\node[right=5pt, darkred] (u) at (4,3) {$u$};

\begin{scope}[ulks,orange!90!black]
\draw[bend right=30] (6) to (7);
\draw[out=-110,in=130] (6) to (8);
\draw[bend right] (4f) to (4b);
\draw (5a) -- (4);
\draw (4a) -- (4c);
\end{scope}

\begin{scope}[ulks, darkred]
\draw (2) -- (4a);
\draw (5) -- (9);
\draw [bend left] (1a) to (1e);
\draw [bend left=30] (1c) to (1b);
\end{scope}

\begin{scope}[lks, green!60!black]
\draw[out=110,in=-175, looseness=1] (7) to (3);
\draw (4f) -- (1e);
\end{scope}

\begin{scope}[lks, blue]
\draw[bend right] (7) to (8);
\draw[bend right=10] (5a) to (4c);
\end{scope}

\begin{scope}[shift={(12,3)}]%
\def\ll{30mm} %
\def\vs{8mm} %

\node[right,green!50!black] at (0,-\vs) {$Y\ \textcolor{black}{\subseteq \delta_L(D_v)}$};
\node[right,blue] at (0,-2*\vs) {$C\ \textcolor{black}{\subseteq L[D_v]}$};
\node[right,orange!90!black] at (0,-3*\vs) {$\Drop_{U[D_v]}(C\cup Y) \ \textcolor{black}{\subseteq U[D_v]}$};
\node[right,darkred] at (0,-4*\vs) {$U\setminus\Drop_{U[D_v]}(C\cup Y)$} ;
\end{scope}

\end{tikzpicture}
 \end{center}
\caption{Illustration of the partial solutions we compute in the dynamic program. The blue links are those in $C$ and the green ones those in $Y$.
Up-links in $U$ are drawn in orange and red, where the orange ones are those in $\Drop_{U[D_v]}(C\cup Y)$.
In this example there is a link $u\in \delta_U(D_v)$. Moreover, the edges in $P_u[D_v]$ are covered by the links in $C\cup Y$.
Thus, in this example, the set $C$ corresponds to the triple $(v,Y,+)$, assuming $k\ge 3$.}\label{fig:table_entries_dp}
\end{figure}
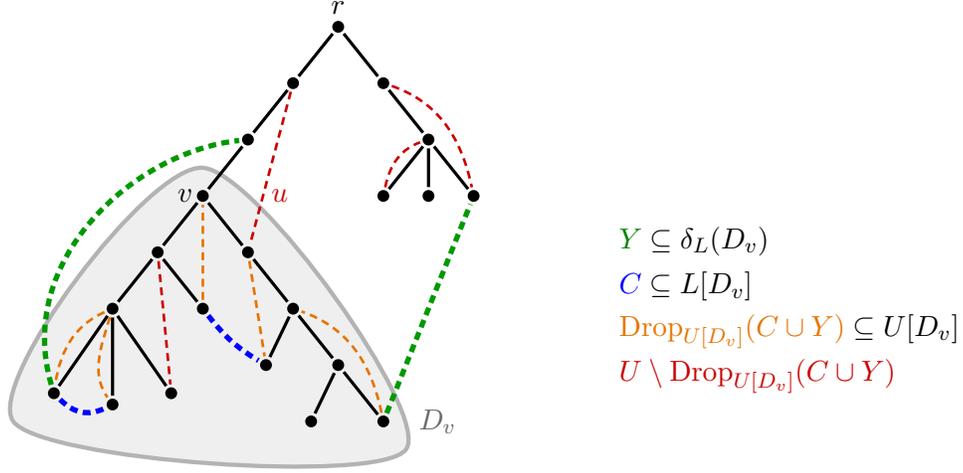

\noindent Our dynamic program will therefore construct link sets $C$ for triples $(v,Y,x)$ consisting of
\begin{enumerate}[label=(\alph*)]
\item\label{item:triplePropV} a vertex $v\in V$,
\item\label{item:triplePropY} a set $Y \subseteq \delta_L(D_v)$ with $|Y|\le k$, and
\item\label{item:triplePropX} $x\in \{+,-\}$.
\end{enumerate}
The value of $x$ being $+$ indicates that there is an up-link $u\in \delta_U(D_v)$ and the triple $(v,Y,x)$ represents a solution $C\cup Y$ that covers all edges of $P_u[D_v]$; otherwise, $x$ should equal $-$. We denote by $\mathcal{T}\subseteq V \times 2^L \times \{+,-\}$ all triples $(v,Y,x)$ fulfilling~\ref{item:triplePropV}--\ref{item:triplePropX}.
We now define formally, when a link set $C\subseteq L[D_v]$ corresponds to the triple $(v,Y,x)$, i.e., it complies with the above-mentioned interpretation of a triple.
\begin{definition}[link set corresponding to $(v,Y,x)$]
Let $(v,Y,x)\in \mathcal{T}$. A link set $C\subseteq L[D_v]$ \emph{corresponds} to the triple $(v,Y,x)$ if $C\cup Y$ is $k$-thin and the following holds. If $x$ equals $+$, then for $C$ to correspond to $(v,Y,x)$ we require that there exists an up-link $u\in \delta_U(D_v)$ and that $P_u[D_v]\subseteq \cup_{\ell\in C\cup Y} P_\ell$.
\end{definition}

We call a triple $(v,Y,x)\in \mathcal{T}$ \emph{feasible} if there exists a link set $C\subseteq L[D_v]$ that corresponds to $(v,Y,x)$. Otherwise, we call $(v,Y,x)$ \emph{infeasible}. By the above definition, this can be rephrased as follows.
\begin{observation}
A triple $(v,Y,x)\in \mathcal{T}$ is \emph{infeasible} if $x$ equals $+$ and either 
\begin{itemize}
\item $\delta_U(D_v) = \emptyset$, or 
\item there is a (single) up-link $u\in \delta_U(D_v)$, but there is no link set $C\subseteq L[D_v]$ such that $C\cup Y$ is $k$-thin and covers $P_u[D_v]$.
\end{itemize}
Otherwise, the triple $(v,Y,x)$ is called \emph{feasible}.
\end{observation}
For each triple $(v,Y,x)\in \mathcal{T}$, our dynamic program decides whether it is infeasible and, if not, computes a set $C\subseteq L[D_v]$ with the following properties:
\begin{enumerate}[topsep=3pt,itemsep=1pt]
\item\label{item:dp1} $C\cup Y$ is $k$-thin;
\item\label{item:dp2} if $x$ equals $+$, in which case there exists a link $u\in \delta_U(D_v)$ because the triple $(v,Y,x)$ is feasible, we have $P_{u}[D_v] \subseteq \bigcup_{\ell\in C\cup Y} P_{\ell}$;
\item\label{item:dp3} $C$ maximizes $\slack_{\rho}(C,Y,v)$ among all link sets $C\subseteq L[D_v]$ satisfying~\ref{item:dp1} and~\ref{item:dp2}. Moreover, if there is a nonempty maximizer, then $C\neq \emptyset$.
\end{enumerate}
See Figure~\ref{fig:table_entries_dp}.
We denote the set $C$ that we compute for a feasible triple $(v,Y,x)\in \mathcal{T}$ by $C(v,Y,x)$. Following standard terminology for dynamic programs, the set $C(v,Y,x)$ is called the \emph{table entry} for the feasible triple $(v,Y,x)$. If $(v,Y,x)\in \mathcal{T}$ is infeasible, then we simply store in the table entry that this triple is infeasible.

Note that being able to efficiently compute, for all feasible triples $(v,Y,x)\in \mathcal{T}$, a link set $C(v,Y,x)$ satisfying the properties~\ref{item:dp1}--\ref{item:dp3}, implies Lemma~\ref{lem:slack_maximization}. Indeed, because $U[D_r]=U$ and $L[D_v]=L$, the set $C=C(r,\emptyset,-)$ maximizes $\slack_{\rho}(C) = \slack_{\rho}(C,\emptyset,r)$ among all $k$-thin sets $C\subseteq L$. Moreover, if there is a nonempty maximizer, then, due to property~\ref{item:dp3}, the computed maximizer $C(r,\emptyset,-)$ is nonempty.

\smallskip

We now discuss how to compute the table entries by starting from the leaves and propagating them up to the root. The table entries for leaves are trivial to compute and hence we focus on the propagation step of the dynamic program.
 More precisely, we discuss how to compute a table entry for a triple $(v,Y,x)\in \mathcal{T}$, assuming that we already computed the table entries for all triples $(v',Y',x')\in \mathcal{T}$, where $v'$ is a child of $v$ in $G$. 

Let $(v,Y,x)\in \mathcal{T}$ and let $v_1,\dots,v_m$ be the children of $v$ in $G$. First, if $x$ equals $+$ and $\delta_U(D_v) = \emptyset$, then the triple $(v,Y,x)$ is infeasible and we save this information as the table entry. Hence, in what follows, assume that there is a (single) link in $\delta_U(D_v)$ if $x$ equals $+$.
To compute the table entry $C(v,Y,x)$, we use the following observation about how any link set $C\subseteq L[D_v]$ that corresponds to the triple $(v,Y,x)$ naturally decomposes into link sets contained in the subtrees of the children of $v$ and constantly many further links. (Think of $C$ as a maximizer $C(v,Y,x)$ we try to find.) 
More precisely, $C\cup Y$ can be partitioned into the sets
\begin{itemize}[topsep=0pt,itemsep=1pt]
\item $C_i \coloneqq C\cap L[D_{v_i}]$ for $i\in \{1,\ldots, m\}$, and
\item $\overline{Y} \coloneqq Y \cup \{\ell \in C \colon v\in V_\ell \}$.
\end{itemize}
Then $C= (\overline{Y} \setminus Y) \cup \bigcup_{i=1}^m C_i$.
See Figure~\ref{fig:combining_partial_solutions}. 
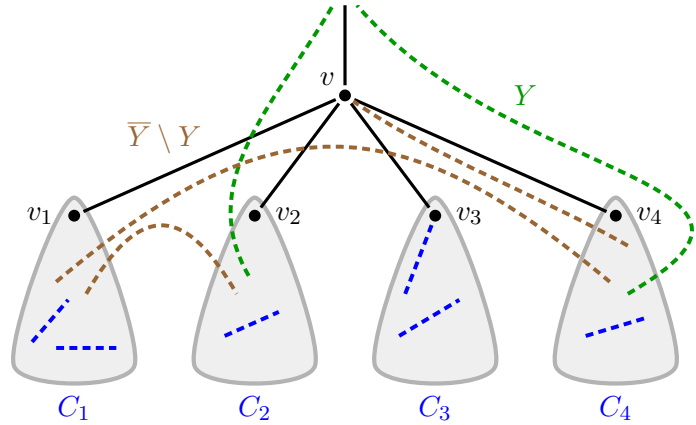
\begin{figure}[!ht]
\begin{center}
\begin{tikzpicture}[scale=0.8]

\tikzset{
lks/.style={line width=1.5pt, densely dashed}
}

\clip (-2,-2.5) rectangle (11,4.5);

\draw[ultra thick, fill=black!20!white, opacity=0.3] plot [smooth cycle] coordinates { (-1,-1.5) (0,1.3) (1,-1.5)};
\draw[ultra thick, fill=black!20!white, opacity=0.3] plot [smooth cycle] coordinates { (2,-1.5) (3,1.3) (4,-1.5)};
\draw[ultra thick, fill=black!20!white, opacity=0.3] plot [smooth cycle] coordinates { (5,-1.5) (6,1.3) (7,-1.5)};
\draw[ultra thick, fill=black!20!white, opacity=0.3] plot [smooth cycle] coordinates { (8,-1.5) (9,1.3) (10,-1.5)};

\begin{scope}[every node/.style={thick,draw=black,fill=black,circle,minimum size=0pt, inner sep=1.3pt, outer sep=2pt}]
\node (0) at (4.5,3) {};
\node (1) at (0,1) {};
\node (2) at (3,1) {};
\node (3) at (6,1) {};
\node (4) at (9,1) {};
\end{scope}

\node[above left] (v) at (0) {$v$};
\node[left=4pt] (v1) at (1) {$v_1$};
\node[right=4pt] (v2) at (2) {$v_2$};
\node[right=4pt] (v3) at (3) {$v_3$};
\node[right=4pt] (v4) at (4) {$v_4$};

\begin{scope}[very thick]
\draw (0) -- (1);
\draw (0) -- (2);
\draw (0) -- (3);
\draw (0) -- (4);
\draw (4.5,4.5) -- (0);
\end{scope}

\begin{scope}[lks, green!60!black]
\draw[out=30, in=-50, looseness=1.5] (9.2,-0.3) to (4.7,4.5);
\draw[out=120, in=-120] (2.9, 0.0) to (4.3,4.5);
\end{scope}

\begin{scope}[lks, brown!80!black]
\draw[in=140,out=40, looseness=1.3] (-0.3,-0.1) to (8.9,-0.1);
\draw[out=60,in=120,looseness=1.8] (0.2,-0.3) to (2.7,-0.3);
\draw[in=-35, out=150, looseness=0.8] (9.2,0.5) to (0);
\end{scope}

\begin{scope}[lks, blue]
\draw (-0.3,-1.2) to (0.7,-1.2);
\draw (-0.7,-1.1)  -- (-0.1,-0.4);
\draw (2.5,-1) -- (3.4,-0.6);
\draw (5.4, -1) to (6.4,-0.4);
\draw (5.5,-0.3) to (3);
\draw  (8.5,-1) -- (9.5,-0.7);
\end{scope}

\node[green!60!black] () at (7.5,3) {$Y$};
\node[brown!80!black] () at (1.5,2.3) {$\overline{Y}\setminus Y$};
\node[blue] () at (0,-2.2) {$C_1$};
\node[blue] () at (3,-2.2) {$C_2$};
\node[blue] () at (6,-2.2) {$C_3$};
\node[blue] () at (9,-2.2) {$C_4$};

\end{tikzpicture}
 \end{center}
\caption{Illustration of the sets $C_i$, $Y$, and $\overline{Y}$.
}
\label{fig:combining_partial_solutions}
\end{figure}

Note that $|\overline{Y}|\le k$ because $\overline{Y}$ is a subset of the $k$-thin set $C\cup Y$ and all links $\ell\in \overline{Y}$ fulfill $v\in V_{\ell}$. 
Thus, $\overline{Y}$ has the following properties:
\begin{itemize}[topsep=3pt,itemsep=1pt]
\item $\overline{Y}\subseteq \{\ell\in L \colon v\in V_{\ell}\}$;
\item $\overline{Y}\cap \delta_L(D_v) = Y$;
\item $|\overline{Y}|\leq k$; 
\item if $x$ equals $+$ and if the link $u\in \delta_U(D_v)$ satisfies $u\not\in \delta_U(v)$, then $u\in \delta_U(D_{v_i})$ for some child $v_i$ of $v$ in $G$ and we have $\overline{Y}\cap \delta_L(D_{v_i}) \neq \emptyset$.
\end{itemize}
Let $\mathcal{Y}$ denote the family of all sets $\overline{Y}\subseteq L$ with these four properties. 
Because $|\overline{Y}|\leq k$ for all $\overline{Y}\in\mathcal{Y}$ and $k$ is constant, the family $\mathcal{Y}$ has only polynomial size. 

So far we have shown that for any set $C$ corresponding to the triple $(v,Y,x)$, the set $C\cup Y$ can be partitioned into sets $\overline{Y}\in \mathcal{Y}$ and sets $C_i \subseteq L[D_{v_i}]$ for $i=1,\dots,m$. 
Thus, to maximize $\slack_{\rho}(C,Y,v)$ over links sets $C\subseteq L[D_v]$ corresponding to the triple $(v,Y,x)$, we can proceed as follows.
First we enumerate over $\overline{Y}$, which can be done efficiently because $\mathcal{Y}$ has only polynomially many elements.
Then we find, for each $\overline{Y}\in \mathcal{Y}$, sets $C_i\subseteq L[D_{v_i}]$ for $i\in \{1,\ldots, m\}$ that maximize $\slack_{\rho}(C,Y,v)$ for $C\coloneqq (\overline{Y}\setminus Y)\cup \bigcup_{i=1}^m C_i$. 

\smallskip

We now discuss how we can find the sets $C_i$ for a fixed $\overline{Y} \in \mathcal{Y}$.
More precisely, we show how to efficiently find sets $C_i\subseteq L[D_{v_i}]$ for $i\in \{1,\ldots, m\}$ such that 
\begin{equation}\label{eq:COverlineY}
C_{\overline{Y}} \coloneqq (\overline{Y} \setminus Y) \cup \bigcup_{i=1}^m C_i
\end{equation} 
is $k$-thin, corresponds to the triple $(v,Y,x)$, and maximizes $\slack_{\rho}(C_{\overline{Y}},Y,v)$ among all such sets $C_{\overline{Y}}$.
If there is a non-empty maximizer, then the set $C_{\overline{Y}}$ we compute is nonempty.
Moreover, if there are no sets $C_i$ such that the resulting set $C_{\overline{Y}}$ as defined in~\eqref{eq:COverlineY} is $k$-thin and corresponds to the triple $(v,Y,x)$, then we will detect this.

This is all that remains to be done, because if $(v,Y,x)$ is feasible, then any set $C_{\overline{Y}}$ that maximizes $\slack_{\rho}(C_{\overline{Y}},Y,v)$ among all $\overline{Y}\in \mathcal{Y}$ is an optimal entry for the triple $(v,Y,x)\in \mathcal{T}$ (where we choose $C_{\overline{Y}}\ne \emptyset$ if a non-empty maximizer exists).
Otherwise, the triple $(v,Y,x)$ is infeasible. We detect this because we cannot find a set $C_{\overline{Y}}$ corresponding to $(v,Y,x)$ for any $\overline{Y}\in \mathcal{Y}$.

\smallskip

To find optimal sets $C_i$ for a fixed $\overline{Y}\in \mathcal{Y}$, we first observe that, for any sets $C_i\subseteq L[D_{v_i}]$ for $i\in \{1,\ldots, m\}$, we have
\begin{equation}\label{eq:slack_extension}
\slack_{\rho}(C_{\overline{Y}}, Y,v) = \sum_{i=1}^m \slack_{\rho}\Bigl(C_i,\ \overline{Y}\cap \delta_L(D_{v_i}),\ v_i\Bigr)  + \rho \cdot \hspace{-2em}\sum_{\substack{i\in I^+:\\ u_i \in \Drop_{U}(C_i\cup \overline{Y})}} \hspace{-2em} w(u_i) - w(\overline{Y} \setminus Y)\enspace,
\end{equation}
where $I^+ \subseteq \{1,\dots,m\}$ is the set of indices $i\in \{1,\ldots, m\}$ for which there is an up-link $u_i\in \delta_U(v) \cap \delta_U(D_{v_i})$.
Moreover, $C_{\overline{Y}}\cup Y$ is $k$-thin if and only if the sets $C_i$ are chosen such that $C_i \cup (\overline{Y}\cap \delta_L(D_{v_i}))$ is $k$-thin for all $i\in\{1,\dots, m\}$.

For each $i\in\{1,\dots,m\}$ let $Y_i\coloneqq (\overline{Y}\cap \delta_L(D_{v_i}))$.
Note that $u_i \in \Drop_U(C_i \cup \overline{Y})$ if and only if $u_i \in \Drop_U(C_i \cup Y_i)$.
Because of~\eqref{eq:slack_extension}, finding optimal sets $C_i$ reduces to finding, for each $i\in \{1,\ldots, m\}$, a set $C_i\subseteq L[D_{v_i}]$ that maximizes
\begin{equation*}
\begin{cases}
\slack_{\rho}(C_i, Y_i, v_i) + \rho\cdot w(u_i) &\text{if $i\in I^+$ and $u_i \in \Drop_{U}(C_i\cup Y_i)$}, \\
\slack_{\rho}(C_i, Y_i, v_i)                    &\text{otherwise}
\end{cases}
\end{equation*}
among all sets $C_i\subseteq L[D_{v_i}]$ for which $C_i \cup Y_i$ is $k$-thin. 
For each $i\in\{1,\dots,m\}$, such a maximizer $C_i$ is obtained as follows.
Let $C_i^-\coloneqq C(v_i,Y_i,-)$. Moreover, if the triple $(v_i,Y_i,+)$ is feasible, then we also define $C_i^+\coloneqq C(v_i,Y_i,+)$.
We choose $C_i$ to be either $C_i^-$ or, if $(v_i,Y_i,+)$ is feasible, possibly $C_i^+$, as described in the below case distinction. (See Figure~\ref{fig:cases_subtrees} for an illustration of the cases.)

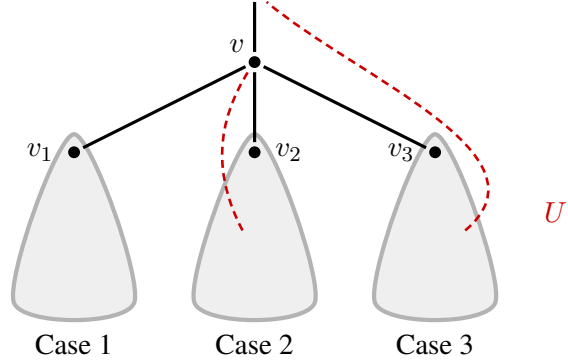
\begin{figure}[!ht]
\begin{center}
\begin{tikzpicture}[scale=0.8]

\tikzset{
lks/.style={line width=2pt, densely dashed}
}

\tikzset{
ulks/.style={line width=1pt, densely dashed}
}

\draw[ultra thick, fill=black!20!white, opacity=0.3] plot [smooth cycle] coordinates { (-1,-1.5) (0,1.3) (1,-1.5)};

\draw[ultra thick, fill=black!20!white, opacity=0.3] plot [smooth cycle] coordinates { (2,-1.5) (3,1.3) (4,-1.5)};

\draw[ultra thick, fill=black!20!white, opacity=0.3] plot [smooth cycle] coordinates { (5,-1.5) (6,1.3) (7,-1.5)};

\begin{scope}[every node/.style={thick,draw=black,fill=black,circle,minimum size=0pt, inner sep=1.3pt, outer sep=2pt}]

\node (0) at (3,2.5) {};
\node (1) at (0,1) {};
\node (2) at (3,1) {};
\node (3) at (6,1) {};

\end{scope}

\node (1a) at (0,-2.2) {Case 1};
\node (2a) at (3,-2.2) {Case 2};
\node (3a) at (6,-2.2) {Case 3};

\node[above left] (v) at (0) {$v$};
\node[left=4pt] (v1) at (1) {$v_1$};
\node[right=4pt] (v2) at (2) {$v_2$};
\node[left=4pt] (v3) at (3) {$v_3$};

\begin{scope}[very thick]
\draw (0) -- (1);
\draw (0) -- (2);
\draw (0) -- (3);
\draw (3,3.5) -- (0);
\end{scope}

\begin{scope}[ulks, darkred]
\draw[bend left] (2.8,-0.3) to (0);
\draw[out=40, in=-40] (6.5,-0.3) to (3.2,3.5);
\end{scope}

\node[darkred] () at (8,0) {$U$};

\end{tikzpicture}
 \end{center}
\caption{Illustration of the three cases for choosing $C_i\in \{ C_i^-,C_i^+\}$.
}
\label{fig:cases_subtrees}
\end{figure}

\begin{description}
\item[Case 1:] The set $\delta_U(D_{v_i})$ is empty.\\[1mm]
In this case we set $C_i \coloneqq C^-_i$.

\item[Case 2:] An up-link $u_i\in \delta_U(D_{v_i}) \cap \delta_U(v)$ exists. (Equivalently, $i\in I^+$.) \\[1mm]
If $Y_i \coloneqq\overline{Y} \cap \delta_L(D_{v_i}) = \emptyset$ or $(v_i, Y_i, +)$ is infeasible, we set $C_i \coloneqq C_i^-$.
Otherwise,
\begin{equation*}
C_i \coloneqq \begin{cases}
C_i^+ & \text{ if }
\slack_{\rho}\Bigl(  C_i^+,\ Y_i,\ v_i\Bigr) + \rho \cdot w(u_i) \ge 
\slack_{\rho}\Bigl(C_i^-,\ Y_i,\ v_i\Bigr)\enspace,
\\
C_i^- &\text{ otherwise}\enspace.
\end{cases}
\end{equation*}

\item[Case 3:] An up-link $u_i\in \delta_U(D_{v_i}) \cap \delta_U(D_v)$ exists.\\[1mm]
If $x$ equals $+$ and the tuple $(v_i,Y_i,+)$ is infeasible, then there are no sets $C_i \subseteq L[D_{v_i}]$ such that $C_{\overline{Y}}$ as defined in~\eqref{eq:COverlineY} corresponds to the triple $(v,Y,x)$.
Otherwise, we set $C_i \coloneqq C^x_i$.
\end{description}

One can easily check that, for a feasible triple $(v,Y,x)$ and fixed set $\overline{Y}\in \mathcal{Y}$, the above choice of $C_i$ leads to a $k$-thin set $C_{\overline{Y}}$ as defined in~\eqref{eq:COverlineY} that corresponds to $(v,Y,x)$ and, among all possible choices for the sets $C_i$, maximizes $\slack_{\rho}(C_{\overline{Y}},Y,v)$. Moreover, if there is no choice of the $C_i$ that leads to a set $C_{\overline{Y}}$ that corresponds to the triple $(v,Y,x)$, then this will be correctly detected in the third case above. Finally, the set $C_{\overline{Y}}$ is nonempty whenever there is a nonempty maximizer, because it is composed of sets $C_i$ that are nonempty whenever possible.

\smallskip

It remains to analyze the running time of the algorithm.
The number of triples $(v,Y,x)$ we consider is no more than $|\mathcal{T}|\leq |V| \cdot |V|^{2k} \cdot 2$. (Note that $|V|^{2k}$ is an upper bound on the number of subset $Y$ of $L\subseteq \left(\!\begin{smallmatrix} V\\ 2 \end{smallmatrix}\!\right)$ of up to $k$ links.) For each of these, the number of sets $\overline{Y}$ we enumerate can be bounded by $|\mathcal{Y}|\leq |V|^{2k}$ because $|\overline{Y}|\leq k$ for all $\overline{Y}\in\mathcal{Y}$. Because the number of children of a vertex $v$ is at most $|V|$, computing $C_{\overline{Y}}$ for a fixed set $\overline{Y}$ (and a fixed triple $(v,Y,x)\in \mathcal{T}$) takes time polynomially bounded in $|V|$.
Thus, the algorithm takes $|V|^{O(k)}$ time altogether, which is polynomial for constant $k$.
This completes the proof of Lemma~\ref{lem:slack_maximization}.
 
\printbibliography

\end{document}